\documentclass{article}
\usepackage{amsmath,amsthm,amsfonts,graphicx,amssymb}
\usepackage{multirow}
\usepackage{physics}
\usepackage{subcaption}
\usepackage{booktabs}
\usepackage{algpseudocode,algorithm,algorithmicx}
\usepackage{mathrsfs}

\newtheorem{theorem}{Theorem}
\newtheorem{lemma}{Lemma}

\newtheorem{definition}{Definition}

\newtheorem{corollary}{Corollary}
\newtheorem{observation}{Observation}

{\theoremstyle{definition}
}
\usepackage{color}
\usepackage[export]{adjustbox}

\usepackage{tikz}
\usepackage{scalerel}
\usetikzlibrary{matrix,calc,shapes,bending} %mpj magically fixes arrow gap problem!
\usetikzlibrary{knots,calc}

\usetikzlibrary{decorations.pathmorphing} % noisy shapes
\usetikzlibrary{fit}					% fitting shapes to coordinates
\usetikzlibrary{backgrounds}	% drawing the background after the foreground
\usetikzlibrary{fit, positioning}

\usepackage{bold-extra}
\usepackage{url}

\newcommand{\ncnod}{\textbf{NC4}}
\newcommand{\apdh}{\textsc{Acyclic B-Hyperpath Existence in a 1-2-Hypergraph}}
\newcommand{\minapdh}{\textsc{Min Acyclic B-Hyperpath in a 1-2-Hypergraph}}
\newcommand{\minapdhprom}{\textsc{Min Acyclic B-Hyperpath in a B-Connected 1-2-Hypergraph}}
\newcommand{\closureprob}{\textsc{Inconsistent Rooted Triple Set Closure}}
\newcommand{\entailprob}{\textsc{Inconsistent Rooted Triple Set Entailment}}

\newcommand{\arcs}{\text{arcs}}
\newcommand{\arc}{\text{arc}}
\newcommand{\triples}{\text{triples}}
\newcommand{\triple}{\text{triple}}

\newif\ifwabi

\wabifalse

\DeclareMathSymbol{\mlq}{\mathord}{operators}{``}
\DeclareMathSymbol{\mrq}{\mathord}{operators}{`'}

\makeatletter
\newcommand\niton{\mathrel{\m@th\mathpalette\canc@l\owns}}
\newcommand\canc@l[2]{{\ooalign{$\hfil#1/\mkern1mu\hfil$\crcr$#1#2$}}}
\makeatother

\begin{document}

\title{Deciding the Closure of Inconsistent Rooted Triples is NP-Complete}

\author{Matthew P. Johnson\thanks{Department of Computer Science, Lehman College and Ph.D. Program in Computer Science, The Graduate Center, City University of New York}
}
%Department
%of Computer Science, Lehman College and the Ph.D.~Program in Computer Science, The Graduate Center, City University of New York
%}
%\affil{Lehman College and The Graduate Center, City University of New York}

\maketitle

\begin{abstract}
% !TEX root = rootedtriples-arx.tex

Interpreting three-leaf binary trees or {\em rooted triples} as constraints yields an entailment relation, whereby binary trees satisfying some rooted triples must also thus satisfy others, and thence a closure operator, which is known to be polynomial-time computable. This is extended to inconsistent triple sets by defining that a  triple is entailed by such a set if it is entailed by any consistent subset of it.

Determining whether the closure of an inconsistent rooted triple set can be computed in polynomial time was posed as an open problem in the Isaac Newton Institute's ``Phylogenetics'' program in 2007. It appears (as \ncnod) in a collection of such open problems maintained by Mike Steel, and it is the last of that collection's five problems concerning computational complexity to have remained open.
We resolve the complexity of computing this closure, proving that its decision version is NP-Complete.

In the process, we also prove that detecting the existence of {\em any} acyclic B-hyperpath (from specified source to destination) is NP-Complete, in a significantly narrower special case than the version whose {\em minimization} problem was recently proven NP-hard by 
Ritz et al. This implies it is NP-hard to approximate
(our special case of) their minimization problem to within {\em any} factor.
%
%
%\vskip 2cm
%
%We investigate the computational complexity of the problem of computing the closure of an inconsistent set of rooted triples (i.e., rooted binary trees on three leaves). 
%
%Given an inconsistent set $R$ of rooted triples (i.e., rooted binary trees on three leaves), its closure is the set of all triples entailed by {\em some} consistent subset of $R$.
%
%
%
%
%Given a consistent set of rooted triples (i.e., rooted binary trees on three leaves, interpreted as constraints), its closure is the set of all triples entailed by it. A rooted triple is entailed by an {\em inconsistent} set of rooted triples if it is entailed by 
%
%Interpreting rooted triples (i.e., rooted binary trees on three leaves) as constraints yields a closure operator: the set of all rooted triples that a given (consistent) set of rooted triples entail. A rooted triple is entailed by a given {\em inconsistent} set of rooted triples if it is entailed by some consistent subset of it.
%
%Its complexity was posed as an open problem during the ``Phylogenetics'' program at the Isaac Newton Institute in 2007, and it appears in a list of such open problems maintained by Mike Steel \cite{steel2010challenges,philogenfinreport} (see also \cite{bryant1995extension,bryant1997building}). It is the last of that collection's Network and Supertree Challenges to have remained unsolved. %\footnote{\url{http://www.math.canterbury.ac.nz/~m.steel/Non_UC/files/research/conjecturesINI.pdf}} 
%We resolve the complexity of this problem, proving that its decision version is NP-Complete.
\end{abstract}

% !TEX root = rootedtriples-arx.tex

\section{Introduction}

We investigate the computational complexity of a problem in which, based on a given collection of relationships holding between the leaves of a hypothetical (rooted) binary tree $T$, the task is to infer whatever additional relationships (of the same form) must also hold between $T$'s leaves as a consequence. Various problems in phylogenetic tree reconstruction involve inference of this kind. The specific relationship form in question here, obtaining between some three leaves $p,q,o$ and denoted $pq|o$, is that of the {\em path between $p$ and $q$} being node-disjoint from the {\em path between $o$ and the root}, or equivalently, of the {\em lowest common ancestor (lca) of $p$ and $q$} not being an ancestor of $o$. This relationship is modeled as a {\em rooted triple}, i.e., the (rooted, full) binary tree on leaves $p,q,o$ in which $p$ and $q$ are siblings, and their parent and $o$ are both children of the root.
%children of the unique non-root internal node, and it and $c$ are children of the root. 
Then $pq|o$ holding in $T$ is equivalent to having the {\em subtree of $T$ induced by $p,q,o$} be homeomorphic to $pq|o$'s corresponding three-leaf binary tree.

The problem of computing the set of all rooted triples entailed by a given triple set $R'$ (its {\em closure} $\overline{R'}$) is known to be polynomial-time computable by, e.g., Aho et al.'s BUILD algorithm \cite{bryant1995extension,aho1981inferring} if $R'$ is {\em consistent}, i.e., if there exists a binary tree satisfying all triples in $R'$.

If a rooted triple set $R$ is inconsistent, then a given triple is said to be entailed by $R$ if it is entailed by {\em any} consistent subset $R' \subset R$. That is, the closure $\overline R$ equals the union of the closures of all $R$'s {\em consistent} subsets. Thus the naive brute-force algorithm for computing $\overline R$ suggested by the definition is exponential-time in $|R|$.

Determining the complexity of the problem of computing $\overline R$ was posed in the Isaac Newton Institute's ``Phylogenetics'' program in 2007 \cite{philogenfinreport}, and it appears (as \ncnod) in a collection of such open problems maintained by Mike Steel \cite{steel2010challenges}. That collection's other four problems concerning computational complexity
%Network and Supertree Challenges 
were all solved by 2009 or 2010, but \ncnod\ has remained open.
%is the last of that collection's  to have remained open. 
We resolve the complexity of computing $\overline R$, proving that it is NP-hard. In particular, we prove that its decision version, i.e., deciding whether a given rooted triple is entailed by $R$, is NP-Complete.

%as an open problem in a collection of open problems maintained by Mike Steel \cite{steel2010challenges,philogenfinreport}
%
%, known as \ncnod, of computing the closure of a set (in generally an inconsistent set) of rooted triples. The complexity of this problem was posed as an open problem in a collection of open problems maintained by Mike Steel \cite{steel2010challenges,philogenfinreport} (see also \cite{bryant1995extension,bryant1997building}). It is the last of that collection's Network and Supertree Challenges to have remained unsolved since 2007. %\footnote{\url{http://www.math.canterbury.ac.nz/~m.steel/Non_UC/files/research/conjecturesINI.pdf}} 
%We prove this problem is NP-hard.

In the process, we also obtain stronger hardness results for a problem concerning {\em acyclic B-hyperpaths}, a directed hypergraph problem that has recently been applied to another computational biology application, but interestingly one unrelated to phylogenetic trees and rooted triples: {\em signaling pathways}, the sequences of chemical reactions through which cells respond to signals from their environment (see Ritz et al.~\cite{ritz2015pathway}).

Specifically, we prove that detecting the existence of {\em any} acyclic B-hyperpath (between specified source and destination) is NP-Complete, in a significantly narrower special case (viz., the case in which every hyperarc has one tail and two heads) than the version whose {\em minimization} problem was recently proven NP-hard by Ritz et al. This immediately implies it is NP-hard to approximate (our special case of) their minimization problem to within {\em any} factor. Moreover, even if we restrict ourselves to feasible problem instances (i.e., those for which there exists at least one such acyclic B-hyperpath), we show that this ``promise problem'' \cite{goldreich2006promise} special case is NP-hard to approximate to within factor $|V|^{1-\epsilon}$ for all $\epsilon>0$.

%In the process, we also prove that detecting the existence of an acyclic B-hyperpath is NP-Complete, already in a much narrower special case (specifically, the case in which every hyperarc has one tail node and two head nodes) than the setting recently proven NP-Complete by 
%Ausiello et al.~\cite{ausielloyz1992optimal} and (more recently) 
%Ritz et al.~\cite{ritz2015pathway}. This also implies that it is NP-hard to approximate (this special case of) 
%%Ausiello et al.'s and 
%Ritz et al.'s minimization problem to within {\em any} multiplicative factor.

%Finally, we show that $\ncnod$~is fixed-parameter tractable with respect to a measure of how ``localized'' the triple set's inconsistency is, specifically the size of the size of the single-component Ahograph that Aho et al.~\cite{aho1981inferring}'s BUILD algorithm fails on.

\vskip .1cm
\noindent \textbf{Related work.}
Inference of new triples from a given set of rooted triples holding in a binary tree was studied by Bryant and Steel \cite{bryant1995extension,bryant1997building}, who proved many results on problems involving rooted triples, as well as quartets, and defined the closure of an inconsistent triple set. The polynomial-time BUILD algorithm of Aho et al.~\cite{aho1981inferring} (as well as subsequent extensions and speedups) can be used to construct a tree satisfying all triples in $R$ (and to obtain the closure $\overline R$), or else to conclude than none exists.

Gallo et al.~\cite{gallo1993directed} defined a number of basic concepts involving paths and cycles in directed hypergraphs, including B-connectivity. Ausiello et al.~\cite{ausielloyz1992optimal} studied path and cycle problems algorithmically in directed hypergraphs and showed, via a simple reduction from \textsc{Set Cover}, that deciding whether there exists a B-hyperpath from specified source to destination with $\le \ell$ hyperarcs is NP-Complete. 
%(Because the hypergraph constructed is acyclic, the proof provides hardness both for the problem formulation with an added constraint requiring the B-hyperpath to be acylic, and for the problem formulation without this constraint.)

Ritz et al.~\cite{ritz2015pathway} recently studied a problem involving ``signaling hypergraphs'', which are directed hypergraphs that can contain ``hypernodes''. They modify Ausiello et al.'s hardness reduction from \textsc{Set Cover} to show that deciding the existence of a length$\le$$\ell$ B-hyperpath is NP-Complete already in the special case of directed hypergraphs each of whose hyperarcs has at most 3 head nodes and at most 3 tail nodes (due to \textsc{Set Cover} becoming hard once sets of size 3 are permitted). Ritz et al.'s hardness proof actually does not use the fact that their problem formulation requires the computed B-hyperpath to be acyclic. Because the entire directed hypergraph they construct is (like Ausiello et al.'s) always acyclic, their proof provides hardness regardless of whether the formulation includes an acyclicity constraint. This constraint is essential to our hardness proof, however, so our result does not rule out the possibility that a B-hyperpath minimization problem formulation {\em without an acyclicity requirement} would be easier to approximate.

%(The hardness proof again works regardless of whether the B-hyperpath is required to be acylic.)

% entailment 

%\cite{bang2012finding,bienstock1991complexity,gallo1993directed,thakur2009linear,bryant1995extension,bryant1997building,adams1986n,hellmuth2017matroid,aho1981inferring,ng1996reconstruction,constantinescu1995efficient,hellmuth2014phylogenetics,ausiello2017directed,italiano1989online,grunewald2007closure,deng2016fast,guillemot2007finding}.

%
% !TEX root = rootedtriples-arx.tex

\section{Preliminaries}

\subsection{Rooted Triples}

\begin{definition}
For any nodes $u,v$ of a rooted binary tree (or simply a {\em tree}):
\begin{itemize}
\item $v {}\le{} u$ denotes that $v$ is a {\em descendent} of $u$ (and $u$ is an {\em ancestor} of $v$), 
%(hereafter, simply {\em descendent}), 
i.e., $u$ appears on the path from $v$ to the root; 
% and $v \ne u$;
%\item 
$v < u$ denotes that $v$ is a {\em proper descendent} of $u$ (and $u$ is a {\em proper ancestor} of $v$), i.e., $v \le u$ and $v \ne u$.
%\footnote{We follow the notation of Aho et al.\ \cite{aho1981inferring}, although the opposite convention, using $v {}>{} u$ to denote $v$ being a proper descendent of $u$ might be more natural with respect to prefix order meet-semilattices.}

% $a \wedge b$
%For any two nodes $a,b$ of a given rooted binary tree, 
\item $uv$
%$(a,b)$ (or simply $a b$) 
denotes their {\em lowest common ancestor (lca)}, i.e., the node $w$ of maximum distance from the root that satisfies $w \ge u$ and $w \ge v$.
\end{itemize}
%or their {\em meet} 
%in that tree. %(Note that $a b$ and $pq$ may refer to the same node even if $\{a,b\} \cap \{p,q\} = \varnothing$.)

\end{definition}

\begin{definition}
\begin{itemize}
\ifwabi
\addtolength{\itemindent}{.2em}
\else
\addtolength{\itemindent}{-1.2em}
\fi
\item A {\em rooted triple} (or simply a {\em triple}) $t = (\{p,q\},o) \in {L \choose 2} \times L$ (with $p,q,o$ all distinct, for an underlying finite leaf set $L$) is denoted by the shorthand notation $pq|o$ and represents the constraint: 
%$(a b {}<{} ac$ and $a b {}<{} bc)$. 
{\em the path from $p$ to $q$ is node-disjoint from the path from $o$ to the root.}

\ifwabi
\addtolength{\itemindent}{-.2em}
\else
\addtolength{\itemindent}{1.2em}
\fi
\item The {\em left-hand side (LHS)} of a triple $pq|o$ is $pq$, and its {\em right-hand side (RHS)} is $o$.

\item $L(T)$ denotes the set of leaves of a tree $T$, and $L(R')$ denotes the set of leaves appearing in any of the triples within a set $R'$, i.e., $L(R') = \bigcup_{pq|o \in R'} \{p,q,o\}$.

\item A tree $T$ with $p,q,o \in L(T)$ {\em displays} the triple $pq|o$ (or, $pq|o$ {\em holds} in $T$) if the corresponding constraint holds in $T$. The set of all triples displayed by $T$ is denoted by $r(T)$.
%if $a b {}<{} ac$ and $a b {}<{} bc$. %(Note that within a given tree, we have either $ac=bc$, $a b=cb$, or $ba=ca$, where the third meet is a descendent of the two identical meets.)
%\end{definition}
%
%\begin{definition}
The set of all trees that display {\em all} triples in $R'$ is denoted by $\langle R' \rangle$.
A set of triples $R'$ is {\em consistent} if $\langle R' \rangle$ is nonempty.
\end{itemize}
\end{definition}
\begin{definition}
\begin{itemize}
\ifwabi
\addtolength{\itemindent}{.2em}
\else
\addtolength{\itemindent}{-1.2em}
\fi
\item For a consistent triple set $R'$, a given triple $t$ (which may or may not be a member of $R'$)  is {\em entailed} by $R'$, denoted $R' \vdash t$,
%\footnote{We adopt the notation $\vdash$ rather than the more commonly used $\vdash$ for reasons explained in the Discussion section.} 
if every tree displaying all the triples in $R'$ also displays $t$, i.e., if $t$ is displayed by every tree in $\langle R' \rangle$. The {\em closure} $\overline{R'}$ is the set of all triples entailed by $R'$, i.e., $\overline{R'} = \{t : R' \vdash t\}$, which can also be defined as $\overline{R'} = \bigcap_{T \in \langle R' \rangle} r(T)$ \cite{bryant1995extension}.
%\end{definition}
%%
%\begin{definition}
%For a consistent triple set $R'$, its {\em closure} $\overline{R'}$ is the set of all triples entailed by $R'$, i.e., $\overline{R'} = \{t : R' \vdash t\}$, which can also be defined as $\overline{R'} = \bigcap_{T \in \langle R' \rangle} r(T)$ \cite{bryant1995extension}.
%%
\ifwabi
\addtolength{\itemindent}{-.2em}
\else
\addtolength{\itemindent}{1.2em}
\fi
\item For an inconsistent triple set $R$, a given triple $t$ (which may or may not be a member of $R$) is {\em entailed} by $R$, again denoted $R \vdash t$, if there exists a consistent subset $R' \subset R$ that entails $t$. The closure $\overline{R}$ is again the set of all triples entailed by $R$, or equivalently the union, taken over every consistent subset $R' \subset R$, of $\overline{R'}$, i.e., $\bigcup_{\text{cons.\:}R' \subset R} \overline{R'}$. 
%Every triple $t \in \overline{R}$ is said to be entailed by $R$, again denoted $R \vdash t$.
%(Note that the closure of a triple set $R$, consistent or not, satisfies $R \subseteq \overline{R} \subseteq {L \choose 2} \times L$.)
\end{itemize}
\end{definition}

\begin{table}[t!]
\caption{Variable name conventions, many of which (also) represent leaves in the triple set $R$ constructed in the reduction. Note that the notation $pq$ (for leaves $p,q$) is used to denote both $\text{lca}(p,q)$ {\em and} the hypergraph node whose outgoing hyperarcs represent triples of the form $pq|o$, i.e., those constraining $\text{lca}(p,q)$ from above.}
\centering % to have the caption near the table

\ifwabi
\begin{tabular}{c @{\hskip .35cm}  l @{\hskip .1cm} c}
\else
\begin{tabular}{c  l  c}
\fi

\toprule
$p,q,p',q',o,o'$ & generic leaf variables, especially in triples' LHSs or RHSs (resp.)	& (leaves)\\
$b_i,b'_i,c_j,d_j$, etc. & particular leaf names	& (leaves)\\
$pq$, etc.	&	lowest common ancestor $\text{lca}(p,q)$ of leaves $p,q$		& (leaf 2-sets)\\
$\alpha,\beta,\gamma$	& leaves of target triple $\alpha\beta | \gamma$	& (leaves)\\
$t$			& \multicolumn{2}{l}{rooted triple, especially of form $p_kq_k|o_k = u_k|o_k$}\\
$R$ or $R'$	& \multicolumn{2}{l}{set of triples, especially inconsistent or consistent (resp.)}\\
$L$ or $L(R)$	& set of leaves or set of leaves appearing in members of $R$ (resp.)			& (leaf sets)\\
$u,u_k,v,v',v_k,v_k'$	& hypergraph nodes, especially tail node or head nodes (resp.)		& (leaf 2-sets)\\
$pq$, etc.	&	hypergraph node corresponding to leaves $p,q$						& (leaf 2-sets)\\
$\alpha\beta, ~c_{m+1}\gamma$		& source and destination nodes (resp.)				& (leaf 2-sets)\\
$a_k$	&	\multicolumn{2}{l}{1-2-hyperarc, especially of form $u_k{\to}\{v_k,v_k'\} = p_kq_k{\to}\{p_ko_k,q_ko_k\}$, with}\\
	& \multicolumn{2}{l}{~\:$k \in [\ell] = \{1,...,\ell\}$ indicating $a_k$'s position in a path $P$ of length $|P|=\ell$}\\
%$i \in [n]$		& SAT variable index in range $\{1,...,n\}$\\
%$j \in [m]$		& SAT clause index in range $\{1,...,m\}$\\
$x_i$	& $i$th SAT variable, with $i \in [n]$\\
$C_j$	& $j$th SAT clause, with $j \in [m]$\\
$x_i, \bar x_i$ or $\tilde x_i$	& literals (positive, negative or either, resp.) of $x_i$\\
$x_i^j, \bar x_i^j$ or $\tilde x_i^j$	& the appearance (positive, negative or either, resp.) of $x_i$ in $C_j$	& (leaves)\\
%$\hat i_w^j$	& the index $i$ of the $w$th variable $x_i$ appearing in $C_j$\\
$x_{\hat w}^j, \bar x_{\hat w}^j$ or $\tilde x_{\hat w}^j$	& the $w$th variable appearance in $C_j$	& (leaves)\\
%$\hat i_j$	& the index $i$ of the {\em some} variable $x_i$ appearing in $C_j$\\
$x_{\hat \cdot}^j, \bar x_{\hat \cdot}^j$ or $\tilde x_{\hat \cdot}^j$	& {\em some} (unspecified) variable appearance in $C_j$	& (leaves)\\
$y_i^j, \bar y_i^j$	& helper leaves in $x_i$ gadget for $x_i^j$ and $\bar x_i^j$ (resp.)	& (leaves)\\
%$z_i^j$	& $b_i$ if $i=-1$ else $b'_i$ if $i=0$ else $y_{i/2}$ if $i$ is even else $x_{(i+1)/2}$	& (leaves)\\
$\tilde z_i^j$	& $j$th element in sequence $b_j,b'_j,\tilde x_i^1,\tilde y_i^1 ..., \tilde x_i^m, \tilde y_i^m,b_{j+1},b'_{j+1}$	& (leaves)\\
%$b_i$ if $i=-1$ else $b'_i$ if $i=0$ else $y_{i/2}$ if $i$ is even else $x_{(i+1)/2}$	& (leaves)\\
$F$			& SAT formula\\
%$i$ & $\triangleq$ & index value for store locations\\
%${T}_{c}$ & $\triangleq$ & A very long description of this specific variable and is needed in the research and looks good when wrapped and aligned to the left.\\
%$TC$ & $\triangleq$ & Total overall cost(\$)\\  
%\multicolumn{3}{c}{}\\
%\multicolumn{3}{c}{\underline{Decision Variables}}\\
%\multicolumn{3}{c}{}\\
%$y_f$ & $=$ & \(\begin{cases}
%1,  & \text{if Supplier located at site $f$ is open} \\
%0,  & \text{otherwise} \end{cases}\)\\
\bottomrule
\end{tabular}
\label{tbl:varnames}
\end{table}

We first state a few immediate consequences of these definitions.

\begin{observation}
\begin{itemize}
\ifwabi
\addtolength{\itemindent}{.2em}
\else
\addtolength{\itemindent}{-1.2em}
\fi
\item It can happen that $pp'=qq'$ even if $\{p,p'\} \cap \{q,q'\} = \varnothing$.

\ifwabi
\addtolength{\itemindent}{-.2em}
\else
\addtolength{\itemindent}{1.2em}
\fi
\item In any given tree $T$ having $p,q,o \in L(T)$, exactly one of $pq|o$, $po|q$, and $qo|p$ holds.

\item $pq|o$ ~iff~ $qp|o$ ~iff~ (path: $p$ to $q$) $\cap$ (path: $o$ to the root) $= \varnothing$ ~iff~ $pq < po = qo$.
%($pq {}<{} po$, $pq {}<{} qo$, and $po=qo$).

\item Equivalently, the {\em 3-point condition} for {\em ultrametrics} \cite{semple2003phylogenetics} holds: for all $p,q,o \in L(T)$, we have $pq < po = qo$ ~or~ $oq < op = qp$ ~or~ $op < oq = pq$.

% $ac=bc$, $a b=cb$, or $ba=ca$, where the third meet is a descendent of the two identical meets.

\item Regardless of whether triple set $R$ is consistent, its closure $\overline{R}$ satisfies $R \subseteq \overline{R} \subseteq {L \choose 2} \times L$, and so $|\overline{R}| = O(|L|^3)$.
\end{itemize}
\end{observation}

We state the problem formally.
%Problem \ncnod~is to compute the closure of a triple set that is (w.l.o.g.) inconsistent:
\vskip .15cm

\noindent
{\bfseries{\scshape{\closureprob}}}\\
%\textsc\textbf{\closureprob}\\
\textsc{Instance}: An inconsistent rooted triple set $R$.\\
\textsc{Solution}: $R$'s closure $\overline R = \{t : R \vdash t\}$.
%\vskip .15cm

By the observation above, computing the closure is equivalent to solving the following decision problem for each of the $O(|L|^3)$ triples $t \in {L \choose 2} \times L$.
\vskip .15cm

%deciding, for each triple $t \in {L \choose 2} \times L$, whether $R \vdash t$.

\noindent
{\bfseries{\scshape{\entailprob}}}\\
\textsc{Instance}: An inconsistent rooted triple set $R$ and a rooted triple $t$.\\
\textsc{Question}: Does $R \vdash t$, i.e., does there exists a consistent triple set $R' \subset R$ satisfying $R' \vdash t$?
\vskip .15cm

Although there is no finite set of inference rules that are complete \cite{bryant1995extension}, there are only three possible inference rules inferring from two triples \cite{bryant1995extension}.

\begin{definition}\label{def:2ord}
The three {\em dyadic inference rules} ($\forall\:p,q,o,p',o' \in L$)  are:
\begin{eqnarray}
\{pq|o, ~ qp'|o\} &\vdash& pp'|o\nonumber\\
\{pq|o, ~ qo|o'\} &\vdash& \{pq|o', ~ po|o'\}\label{eq:dyadic}\\
\{pp'|o, ~ oo'|p\} &\vdash& \{pp'|o', ~ oo'|p'\}\nonumber
\end{eqnarray}
\end{definition}

%We also introduce the following artificial definition, which will be useful later.
%
%\begin{definition}
%We say that a triple is {\em pseud o-entailed} by a (possibly inconsistent) set $R$ if it can be derived from $R$ through repeated application of the second 2-order inference rule.
%\end{definition}

A type of graph (distinct from hypergraphs discussed below) that will be used in the hardness proof is the {\em Ahograph} \cite{aho1981inferring}, which is defined for a given triple set $R$ and leaf set $L$.\footnote{We choose to define the Ahograph as a multigraph whose edges each have exactly one label, rather than the more common definition as a graph whose edges each have a {\em set} of labels.}

\begin{definition}
For a given triple set $R$ and leaf set $L$, the {\em Ahograph} $[R,L]$ is the following undirected {\em edge-labeled} graph:
\begin{itemize}
\item its vertex set equals $L$;
\item for every triple $pq|o \in R$, if $p,q,o \in L$, then there exists an $\{p,q\}$ with label $o$.
\end{itemize}

For a hypergraph $(V,A)$,  the {\em corresponding Ahograph} is the Ahograph $[\triples(A),V]$.

To avoid confusion with the nodes of the hypergraph, we refer to the Ahograph's nodes and edges 
%({\em which are also leaves of the resulting tree}) 
as {\em A-nodes} and {\em A-edges}.
\end{definition}

\subsection{Directed Hypergraphs}

Definitions of paths and cycles in hypergraphs are subtler and more complicated than the corresponding definitions for graphs (see \cite{nielsen2001remark}).
We adopt versions of Gallo et al.\ \cite{gallo1993directed}'s definitions,
%Ausiello et al.\ \cite{ausielloyz1992optimal} 
%and Ritz et al.\ \cite{ritz2015pathway}.
simplified for the special case in which every hyperarc has exactly one tail and two heads.

\begin{definition}
A {\em 1-2-directed hypergraph} (or simply {\em hypergraph}) $H=(V,A)$ consists of a set of nodes $V$ and a set of 1-2-hyperarcs $A$. A 1-2-hyperarc (or 1-2-directed hyperedge\footnote{Called a $2$-directed F-hyperarc in \cite{thakur2009linear}, extending definitions introduced by Gallo et al.\ \cite{gallo1993directed}.}, or simply {\em hyperarc} or {\em arc}) is an ordered pair $a = (u,\{v,v'\}) \in V \times {V \choose 2}$, with $u,v,v'$ all distinct, which we denote by $u{\to}\{v,v'\}$.
%$S_1 \ne \varnothing$, $S_2 \ne \varnothing$, and $S_1 \cap S_2 = \varnothing$. 
Let $\text{t}(a)=u$ be $a$'s {\em tail} and $\text{h}(a)=\{v,v'\}$ be $a$'s {\em heads}. A node with out-degree 0 is a {\em sink}.
%
%%A {\em $k$-hyperdigraph} is one in which for every hyperarc $a$, $\max\{\text{t}(a), \text{h}(a)\} \le k$ for some positive integer $k$.
%
%A {\em 1-2-hyperdigraph} is one in which for every hyperarc $a$ (which we call a {\em 1-2-hyperarc} or {\em 1-2-arcs} in this case\footnote{Called a $2$-directed F-hyperarc in \cite{thakur2009linear}, extending definitions introduced by Gallo et al.\ \cite{gallo1993directed}.}), we have $|\text{t}(a)|=1$ and $|\text{h}(a)|=2$, i.e., $a \in V \times {V \choose 2}$.
%%Let a {\em 1-2-hyperdigraph} be a directed hypergraph in which all arcs are 1-2-hyperarcs.
%We sometimes write $u{\to}v$ or say that $u$ {\em points to} $v$ when $\{u\} = \text{t}(a)$ and $v \in \text{h}(a)$ for such an arc $a$. Also, we often abuse notation and treat $t(a)$ as being equal to $u$ rather than $\{u\}$.
\end{definition}

\begin{definition}
\begin{itemize}
\ifwabi
\addtolength{\itemindent}{.2em}
\else
\addtolength{\itemindent}{-1.2em}
\fi
\item A {\em simple path} from $u_0$ to $u_\ell$ is a sequence of distinct 1-2-hyperarcs $P=(a_1, ..., a_\ell)$, where $u_0 = \text{t}(a_1)$, $u_\ell \in \text{h}(a_\ell)$ and $\text{t}(a_{k+1}) \in \text{h}(a_k)$ for all $k \in [\ell-1]$.
%For a simple path $P$, we also consider $\arcs(P)=P \cap A$ a simple path. 
The length $|P|=\ell$ is the number of arcs.
%A {\em simple path} from $u_0$ to $u_\ell$ is an alternating sequence of nodes and arcs $(u_0, a_1, u_1, a_2, ..., a_\ell, u_\ell)$, where $u_k = \text{t}(a_k)$ and $u_{k+1} = \text{t}(a_{k+1}) \in \text{h}(a_k)$ for all $k \in [\ell]$, with no arcs repeated.
%For a simple path $P$, we also consider $\arcs(P)=P \cap A$ a simple path. The length $|P|$ of a simple path is its number of arcs.
%A path is {\em simple} if no 
%%nodes or 
%arcs are repeated.
 
\ifwabi
\addtolength{\itemindent}{-.2em}
\else
\addtolength{\itemindent}{1.2em}
\fi
\item A {\em cycle} is a simple path having $\text{h}(a_\ell) \ni \text{t}(a_1)$. %, i.e., $\text{h}(a_\ell) \cap \text{t}(a_1) \ne \varnothing$.
%, and, in particular, it is a {\em simple cycle} if $u_k=u_1$.\footnote{Note that this definition is {\em not} analogous to the usual definition of a cycle being simple in the sense of being minimal.}
An arc $a_k \in P$ having one of its heads be the tail of some earlier arc $a_{k'}$ of $P$, i.e., where $\exists a_{k'} \in P:~  k'<k \text{ and } \text{h}(a_k) \ni \text{t}(a_{k'})$, is a {\em back-arc}.
A simple path is {\em cycle-free} or {\em acyclic} if it has no back-arcs, and is {\em cyclic} otherwise. More generally, a set $A' \subseteq A$ is {\em cyclic} if it is a superset of some cycle, and {\em acyclic} otherwise.
\end{itemize}
\end{definition}

%\begin{remark}
%Note that the definition {\em in general directed hypergraphs} of a simple path from $u_0$ to $u_\ell$ being a cycle does not imply that $u_0 = u_\ell$, only that $u_0 \in \text{h}(a_\ell)$. Since all the arcs in our construction will be 1-2-arcs, however, every cycle will satisfy $u_1=u_k$.
%\end{remark}

%\begin{definition}
%Let $\arcs(P) = P \cap A$ denote the arcs appearing in a path $P$. We also call $\arcs(P)$ a {\em cycle} if $h(a_1) \in t(a_k)$, and we say that $\arcs(P)$ (or more generally any subset of $A$) is {\em cyclic} if $\arcs(C) \subseteq \arcs(P)$ for some cycle $C$, and otherwise it is {\em acyclic}.
%\end{definition}

%
%A subpath of a $u_1{-}u_k$ path is a contiguous subsequence $u_i, a_i, u_{k+1}, a_{k+1}, ..., a_{j-1}, u_j$ of it for some $1 \le i \le j \le k$.
%
%A {\em cycle} is a $u_1{-}u_k$ path in which $u_k \in \text{t}(a_1)$.
%%
%A path is {\em acyclic} if no subpath of it is a cycle.
%
%Let a {\em path} in $H$ be a sequence of arcs $a_{k_1},a_{k_2},...,a_{k_k}$ such that each arc's tail is one of the previous arc's two heads, i.e., for each $2 \le h \le k$, we have $\text{t}(a_{k_h}) \in \text{h}(a_{k_{h-1}})$. 

\begin{definition}
In {\em general} directed hypergraphs (i.e., with no restrictions on arcs' numbers of heads and tails), a node $v$ is {\em B-connected}\footnote{Note also that Gallo et al.\ \cite{gallo1993directed} defines {\em B-hyperarc} simply to mean an arc $a$ having $|h(a)|=1$.} to $u_0$ if $v=u_0$ or (generating such B-connected nodes bottom-up, through repeated application of this definition) if there is a hyperarc $a$ with $v \in \text{h}(a)$ and every node $\text{t}(a)$ is B-connected to $u_0$.
%Gallo et al.\ \cite{gallo1993directed}'s definition also requires that every node in the path be connected to $s$ by a {\em cycle-free path}, i.e., a path containing no subpath that is a (non-simple) cycle, but Ritz et al.\ \cite{ritz2015pathway}'s definition does not require this. Because in our construction every hyperedges will have one tail, all cycles will be simple, and hence our acyclic hyperpath hardness result applies to both definitions. 
A path $P$ from $u_0$ to $u_\ell$ is a {\em B-hyperpath} if $u_\ell$ is B-connected to $u_0$ (using only the arcs $a \in P$).
\end{definition}

{\em Due to the following observation, for the remainder of this paper any use of the term ``path'' will be understood to mean ``B-hyperpath''.}
\begin{observation}
If all arcs are 1-2-hyperarcs, then every simple path is also a B-hyperpath.
\end{observation}
%\vskip -.1cm

Via the hypergraph representation used in our hardness proof for \entailprob\ below, we also obtain hardness results for the following problem formulations as a by-product.
\vskip .15cm

\noindent
{\bfseries{\scshape{\apdh}}}\\
\textsc{Instance}: A 1-2-directed hypergraph $H=(V,A)$ and nodes $u,v \in V$.\\
\textsc{Question}: Does there exist an acyclic B-hyperpath in $H$ from from $u$ to $v$?
\vskip .15cm

%This immediately implies that approximating the following optimization problem to within any multiplicative factor is NP-hard.
%\vskip .05cm

We want to define an optimization version of the problem where the objective is to minimize path $P$'s length $|P|$, but since a given problem solution may contain no solutions at all (it may be {\em infeasible}, specifically if $v$ is not B-connected to $u$), we obtain the following somewhat awkward definition. Note that defining the cost of an infeasible solution to be infinity is consistent with the convention that $\min \varnothing = \infty$.

\vskip .15cm

\noindent
{\bfseries{\scshape{\minapdh}}}\\
\textsc{Instance}: A 1-2-directed hypergraph $H=(V,A)$ and nodes $u,v \in V$.\\ %, where the $v$ is B-connected to $u$.\\
\textsc{Solution}: A B-hyperpath $P$ in $H$.\\ %An acyclic B-hyperpath $P$ in $H$ from $u$ to $v$.\\
\textsc{Measure}: $P$'s length $|P|$, (i.e., its number of hyperarcs), if $P$ is a {\em feasible} solution (i.e., an acyclic B-hyperpath from $u$ to $v$), and otherwise infinity.
\vskip .15cm

%Note that if $v$ is not B-connected to $u$ in a given instance of \minapdh\, then the instance will have no solutions at all (it is {\em infeasible}). Under the convention that $\min \varnothing = \infty$, an infeasible instance's optimal solution cost is infinity. 

Alternatively, we can formulate a ``promise problem''~\cite{goldreich2006promise} special case of the minimization problem, restricted to 
%instances in which $v$ is B-connected to $u$, i.e., 
instances admitting feasible solutions.

\vskip .15cm
\noindent
{\bfseries{\scshape{\minapdhprom}}}\\
\textsc{Instance}: A 1-2-directed hypergraph $H=(V,A)$ and nodes $u,v \in V$, where the $v$ is B-connected to $u$.\\
\textsc{Solution}: An acyclic B-hyperpath $P$ in $H$ from $u$ to $v$.\\
\textsc{Measure}: $P$'s length $|P|$. % (i.e., its number of hyperarcs).
\vskip .15cm

\section{The Construction}

\subsection{High-level Strategy}

We will prove that \entailprob\ is NP-Complete by reduction from \textsc{3SAT}, using a construction similar to that of \cite{bang2012finding} (see also \cite{bienstock1991complexity}) for the problem of deciding whether a specified pair of nodes in a directed graph are connected by an {\em induced} path.\footnote{That problem becomes trivial if either the graph is undirected or the {\em induced} constraint is removed.} So, given a SAT formula $F$, we must construct a problem instance $(R,t)$ such that $R \vdash t$ iff $F$ is satisfiable. Intuitively, we want to define $R$ in such a way that it will be representable as a graph (or rather, as a directed hypergraph), whose behavior will mimic that of the induced subgraph problem.

%construct a graph encoding $R$ and $t$ that will play an intermediate role between $(R,t)$ and $F$, so that $R$ entailing $t$ corresponds to the existence of a certain kind of path, which corresponds in turn to $F$ being satisfied.

%where the existence of a certain kind of path corresponds to $t$ being entailed, and the SAT formula being satisfied.

%\subsection{Mimicking Induced-Subgraph Consequences}

In slightly more detail, the instance $(R,t)$ that we define based $F$ will have a structure that makes it representable as a certain directed hypergraph. This hypergraph (see Fig.\ \ref{fig:overview}) will play an intermediate role between $(R,t)$ and $F$, yielding a two-step reduction between the three problems.
%, in which
%%
%an acyclic path $P$ (between a specified source and destination) will be the witness (if one exists) verifying $R \vdash \alpha\beta|\gamma$, if $F$ formula is satisfiable. 
In particular, we will show:
%
%\begin{enumerate}
%\item A (possibly cyclic) path $P$  (between a specified source and destination) determines a (possibly non-satisfying) truth assignment $\mathbf{v}(\cdot)$ for the SAT formula $F$, and vice versa.
%\item $P$ will be acyclic iff $\mathbf{v}(\cdot)$ satisfies $F$.
%\item The (possibly cyclic) path $P$ also determines a (minimal) subset $R' \subset R$ (pseud o-)entailing $\alpha\beta|\gamma$, and vice versa.
%\item $P$ will be acyclic
%%(in a sense to be defined) 
%iff $R'$ is consistent.
%\item Hence $R'$ will be consistent iff $P$ is acyclic iff $\mathbf{v}(\cdot)$ satisfies $F$.
%\end{enumerate}
%

\begin{enumerate}
\item A path $P$ (from $\alpha\beta$ to $c_{m+1}\gamma$) determines a truth assignment $\mathbf{v}(\cdot)$, and vice versa.
\item $P$ will be acyclic iff $\mathbf{v}(\cdot)$ satisfies $F$.
\item An acyclic path $P$ (or an acyclic superset of it) determines a consistent subset $R' \subset R$ entailing $t=\alpha\beta|\gamma$, and vice versa.
\item Hence $R'$ will be consistent and entail $\alpha\beta|\gamma$ iff $P$ is acyclic iff $\mathbf{v}(\cdot)$ satisfies $F$.
\end{enumerate}

The challenge we face is designing a construction that will force cycles to autonomously result from non-satisfiable formulas (mimicking the logic of an {\em induced} subgraph) is that the definition of entailment of a triple $t$ from an inconsistent set $R$ allows us to pick and choose among the members of $R$, selecting any consistent subset as the witness to $t$'s entailment, seemingly indicating that any troublesome members of $R$ corresponding to back-arcs causing a cycle could simply be omitted---{\em independently} of our choices selecting the triples that we {\em are} relying on.

The way we disallow this freedom is that we model a rooted triple not as a directed edge in a graph but as a directed {\em hyperedge}, pointing from one tail node to two head nodes. Although the definition of entailment from an inconsistent triple set $R$ means we can omit any hyperarc we like in defining a possible $H'$, we cannot omit {\em half} a hyperarc: ``turning on'' a 1-2-hyperarc $u{\to}\{v,v'\}$ because we want tail $u$ to point to head $v$ also necessarily causes $u$ to point to $v'$.

For most of the arcs we define in our construction, these second head nodes will be just spinning wheels: sink nodes having no effect, and omitted for clarity from some figures. The important ones are those in which tail $u$ and one head $v$ both lie in a clause gadget and the other head $v'$ lies in a variable gadget.

\begin{figure}
\centering
%\captionsetup{justification=centering}
%\input{triples-fig-overview-straight}
\makebox[\textwidth][c]{\hskip 0cm\includegraphics[width=1.1877\textwidth]{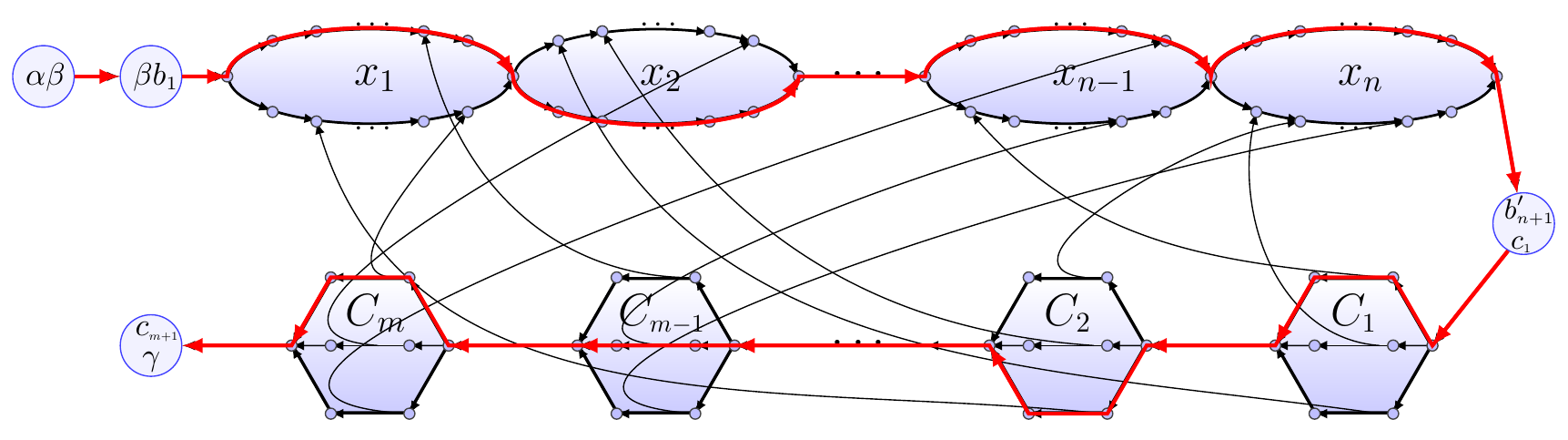}}
\vskip -.05cm
\caption{Construction overview, with the path $P$ from $\alpha\beta$ to $c_{m+1}\gamma$ shown in red. Each ellipse represents the gadget for one variable $x_i$ (see Fig.~\ref{fig:var}), and each hexagon represents the gadget for one clause $C_j$ (see Fig.~\ref{fig:clause}). (Sink nodes are omitted for clarity.)
%Here $x_{n-1}$ is false, which satisfies $C_m$: 
%
The path shown corresponds to a truth assignment in which $x_2$ is true and $x_1,x_3,x_4$ are false.
For example, the path shown takes $x_1$'s {\em positive} (upper) side, passing through its {\em positive} nodes, which {\em renders $x_1$'s positive appearances unusable}, thus setting $x_1$ to {\em false}.
$C_m$'s upper {\em witness path} points to $x_1$'s {\em negative} (lower) side, indicating that $x_1$'s appearance in $C_m$ is {\em negative}. Thus $x_1$ being false satisfies $C_m$.}
\label{fig:overview}
%\vskip .1cm
%
%%\end{figure}
%%\begin{figure}[t!]
%\centering
%%\captionsetup{justification=centering}
%%\includegraphics[width=\linewidth, clip=true, trim={0cm 10cm 3cm 5cm}]{triplescaterpillar.jpg}
%%\input{triples-fig-tree}
%\includegraphics[width=.61\textwidth]{triples-fig-tree.pdf}
%%\vskip .05cm
%\vskip -.1cm
%\caption{The caterpillar tree resulting from the triples corresponding to $P$, which thence satisfies $\alpha\beta|\gamma$.}
%\label{fig:caterpillar}
\end{figure}

\subsection{Identifying Rooted Triples and 1-2-Hyperarcs}

A core idea of our construction and proof is a correspondence between rooted triples and $H$'s hyperarcs (all 1-2-hyperarcs), which renders them mutually definable in terms of one anther. 
%Each hyperarc of the hypergraph $H$ we construct will be a 1-2-hyperarc.
%, pointing from one tail to two heads. 
Each of $H$'s nodes will be identified with an unordered pair of leaves $\{p,q\} \in {L \choose 2}$ (written for convenience $pq$), and each of its hyperarcs will have structure of the form $pq{\to}\{po,qo\}$, with $p,q,o$ all distinct. That is, {\em each of an arc $u{\to}\{v,v'\}$'s two heads $v,v'$ will contain one of the tail $u$'s two leaves plus a different leaf common to both $v$ and $v'$}. This structure ensures that each hyperarc encodes a rooted triple, 
%in this case, $pq | o$, 
rather than a constraint of the more general form $pp' < qq'$
%, a correspondence we (recalling Defs.\ \ref{def:corresp},\ref{def:abstractcorresp}) summarize as follows. 
%That is, there will be a one-to-one correspondence between the (inconsistent) triple set $R$ we define and the set of hyperarcs $A$ we define.
Thus we can write $A = \{ pq{\to}\{po, qo\} : pq | o \in R\}$ or $R = \{ pq | o : pq{\to}\{po, qo\} \in A \}$. Indeed, we can simply {\em identify them with one another} as follows.

\begin{definition}\label{def:corresp}
For a triple $pq|o$, the corresponding hyperarc is $\arc(pq|o) = pq{\to}\{po,qo\}$; conversely, for a 1-2-hyperarc $pq{\to}\{po,qo\}$, the corresponding triple is $\triple(pq{\to}\{po,qo\}) = pq|o$. 
For a triple set  $R'$, we write $\arcs(R')$ to denote the same set $R'$, with but its members {\em treated as arcs}, and similarly in reverse, for an arc set $A'$, we write $\triples(A')$.
\end{definition}

Given %that each 1-2-hyperarc $u{\to}\{v,v'\}$ will be of the form $pq{\to}\{po,qo\}$, 
this,
we can also give a more abstract correspondence.

\begin{definition}\label{def:abstractcorresp}
For a 1-2-hyperarc $u{\to}\{v,v'\}$, the corresponding triple is $\triple(u{\to}\{v,v'\}) = v \oplus v' | v \cap v'$, where $\oplus$ denotes symmetric difference.
We also combine the two models' syntax, writing $u|o$ to denote $pq|o$ when $u=pq$, i.e., when hyperarc $u{\to}\{v,v'\} = \arc(pq|o)$.
\end{definition}

This leads to the following equivalent restatements of the second dyadic inference rule (recall Def.\ \ref{def:2ord}) in forms that will sometimes be more convenient.

\begin{observation}
The first inference of dyadic inference rule (\ref{eq:dyadic}) can be stated as:
\begin{alignat}{3}
\{pq{\to}\{po,qo\},~ qo{\to}\{qo',oo'\}\} ~&\vdash~ pq{\to}\{po',qo'\} ~~~~&&(\forall\:p,q,o,o' \in L)\nonumber\\%\label{eq:2a}\\
\{u_{k-1} | o, ~ u_k | o'\} ~&\vdash~ u_{k-1} | o' ~~~~&&(\forall\:u_{k-1},u_k \in V, o \in u_k \text{ s.t. } |u_{k-1} \cap u_k|=1) \label{eq:2b}
\end{alignat}
%
%
%\begin{enumerate}
%\item[2a.] $\{a b{\to}\{bc,ac\}, bc{\to}\{bd,cd\}\} ~\vdash~ a b{\to}\{ad,bd\}$
%%\end{enumerate}
%%or alternatively as:
%%\begin{enumerate}\setcounter{enumi}{1}
%%\item $\{u_{k-1}{\to}\{v_{k-1},v_{k-1}'\}, ~ v_{k-1}{\to}\{v_k,v_k'\}\} ~\vdash~ u_{k-1}{\to}\{u_{k-1} \oplus v_k,v_k\}$
%%\end{enumerate}
%%or alternatively as:
%%\begin{enumerate}\setcounter{enumi}{1}
%\item[2b.] $\{u_{k-1} | c, ~ u_k | d\} ~(\text{s.t.} ~ c \in u_k \text{ and } |u_{k-1} \cap u_k|=1) ~\vdash~ u_{k-1} | d$
%\end{enumerate}
%%or alternatively as:
%%\begin{enumerate}\setcounter{enumi}{1}
%%\item $\{u_{k-1}{\to}\{u_{k-1} \oplus v_k \oplus v_k', v_k \oplus v_k'\}, ~ v_k \oplus v_k'{\to}\{v_k,v_k'\}\} ~\vdash~ u_{k-1}{\to}\{u_{k-1} \oplus v_k,v_k\}$
%%\end{enumerate}
\end{observation}

\begin{figure}
	\centering %\hskip -,25in
	\begin{subfigure}{\textwidth}
\centering
\makebox[\textwidth][c]{\includegraphics[width=1.05\textwidth]{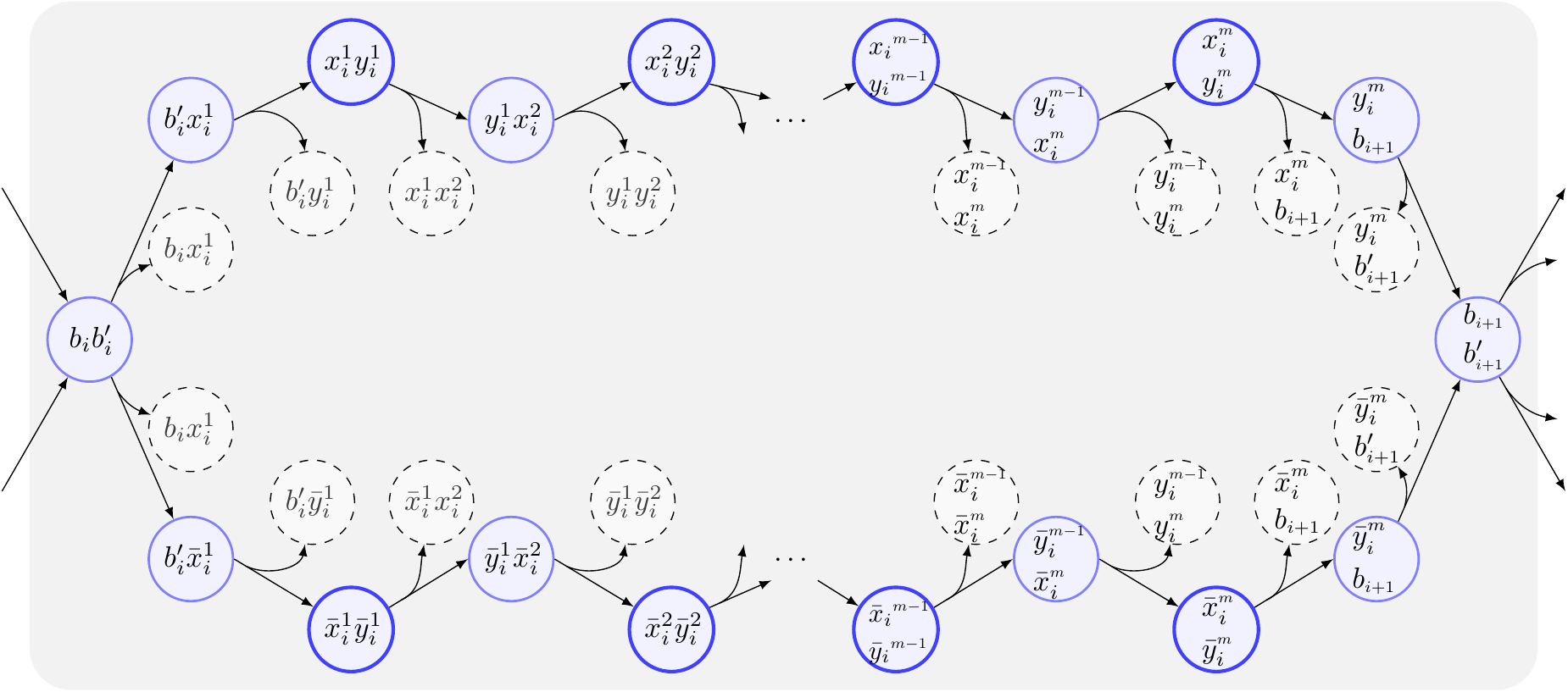}}
		\caption{Variable gadget for $x_i$. Any path passing through this gadget (drawn left to right) has two options, taking its {\em negative} (lower) side, making $x_i$ {\em true}, or its {\em positive} (higher) side, making $x_i$ {\em false}. 
That is, the truth value corresponding to the path is the one making {\em the literals in the nodes on the unused side} true.
Intuitively, {\em a path traversing one of the gadget's two sides renders all the literals appearing within that side's nodes unusable}. Note that the rightmost node ($b_{i+1}b'_{i+1}$) is also (for each $i<n$) the leftmost node of $x_{i+1}$'s gadget.}\label{fig:var}
	\end{subfigure}\\
	\vskip .5cm
	
%%%%%%%%%%%%%%
	\begin{subfigure}{\textwidth}
\centering
\makebox[\textwidth][c]{\includegraphics[width=.61\textwidth]{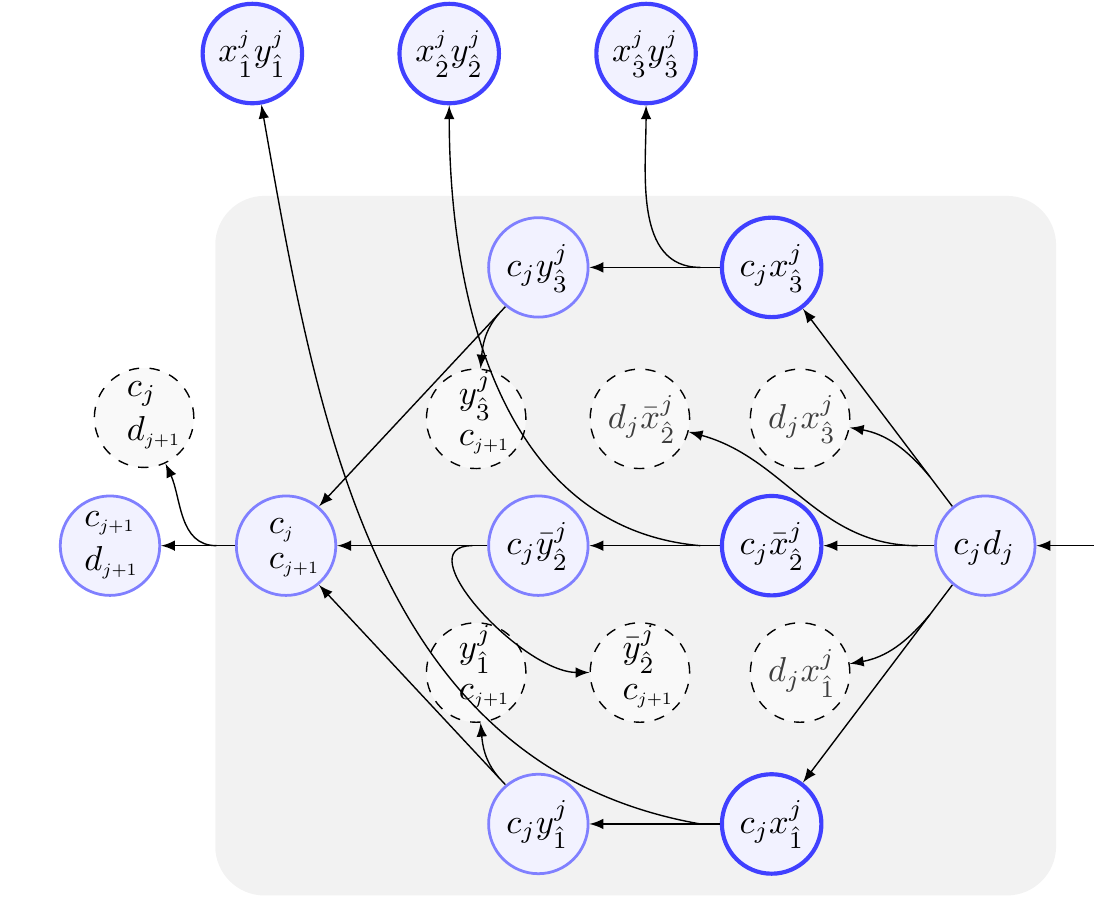}}
		\caption{Clause gadget for $C_j = (x_{i_1}^j \vee \bar x_{i_2}^j \vee x_{i_3}^j)$, which is followed (drawn outside the shaded region) by node $c_{j+1}d_{j+1}$ (or $c_{m+1}\gamma$, in the case of $j=m$). Any path passing through this gadget (drawn right to left) has three options:
%when traveling from $c_jd_j$ to $c_jc_{j+1}$: 
going {\em up}, {\em straight across}, or {\em down}, each corresponding to one choice among $C_j$'s three possible {\em witness paths}.
%(satisfying $C_j$). 
The arrow from the witness path's {\em witness node}, say, $c_j \tilde x_i^j$, to a node $\tilde x_i^j \tilde y_i^j$ lying within {\em one of the two sides} of $x_i$'s gadget (and outside the shaded region) represents {\em the appearance of $x_i$ in $C_j$}; the 1-2-hyperarc that arrow is constituent of forces an acyclic path taking this witness path to have taken the {\em opposite side} of $x_i$'s gadget.}\label{fig:clause}
	\end{subfigure}
%%%%%%%%%%%%%%
	\caption{Gadgets used in the reduction. Each pair of arrows drawn forking from the same tail node represents one 1-2-hyperarc. Sink nodes have dashed borders and are shaded lighter (gray) than non-sink nodes (blue). The clause gadget nodes that point to variable gadget nodes and the variable gadget nodes that can be pointed to by them are both drawn with thick borders.}
\end{figure}

%Each hyperarc is therefore identified with a triple in $R$, i.e., $A = \{ pq{\to}\{po, qo\} \::\: pq | o \in R\}$.

%\begin{fact}
%%Each arc $pq{\to}\{po, qo\} \in A$ of $H$ is equivalent to, and can be identified with, the triple $pq | o \in R$. Similarly, any subset $A' \subset A$ 
%Arcs in $A$ and triples in $R$ can be identified with one another, and mutually definable in terms of one another, i.e. we can write $A = \{ pq{\to}\{po, qo\} \::\: pq | o \in R\}$ or $R = \{ pq | o \::\: pq{\to}\{po, qo\} \}$.
%\end{fact}

We emphasize again the following two related facts about the meaning of an arc $pq{\to}\{po, qo\} \in A$:

\begin{enumerate}
\item If $T$ is a tree with $p,q,o \in L(T)$ and $p q | o \in r(T)$, then lowest common ancestors $po$ and $qo$ are equal, i.e., they refer to the same node in $T$.
\item Yet $po$ and $qo$ are two distinct A-nodes (in $V$) of the hypergraph $H$.
\end{enumerate}

That is, ``turning on'' triple $p q | o$ (by adding it to the triple set $R'$) has the effect of causing the {\em hypergraph nodes} $po$ and $qo$ to thence refer to the same {\em tree node} (in any tree displaying $R'$).

\subsection{Defining $L$ and $R$}

%\subsection{The Nature of Nodes and Arcs in $H=(V,A)$}

Let the SAT formula $F$ on variables $x_1,...,x_n$ consist of $m$ clauses $C_j$, each of the form $C_j = (\tilde x_{i_1}^j \vee \tilde x_{i_2}^j \vee \tilde x_{i_3}^3)$ or $C_j = (\tilde x_{i_1}^j \vee \tilde x_{i_2}^j)$,
%(add a duplicate literal in the case of a clause with only two literals), 
where each literal $\tilde x_i^j$ has the form either $x_i$ or $\bar x_i$ for some $i$.

We define the leaf set $L$ underlying $R$ as $L = L_1 \cup L_2 \cup L_3 \cup L_4$, where:

%we define will includes the following:
\begin{itemize}
\item \makebox[5.5cm]{$L_1 = \bigcup_{i \in [n],j \in [m]} \{x_i^j,\bar x_i^j,y_i^j,\bar y_i^j\}$\hfill} ($4nm$ leaves)\footnote{Alternatively, we could create such nodes only corresponding to actual appearances of variables in clauses, i.e., 
% can simply create characters corresponding to all $4nm$ possible literal/clause pairings: 
$L_1 = \bigcup_{i,j: x_i \in C_j} \{x_i^j,y_i^j\}\:\cup\:\bigcup_{i,j: \bar x_i \in C_j} \{ \bar x_i^j,\bar y_i^j\}$ \: ($\le 3m$ leaves).}
%the $3m$ characters $\bigcup_j \{\tilde x_j^1, \tilde x_j^2, \tilde x_j^3\}$, 
\item \makebox[5.5cm]{$L_2 = \bigcup_{i \in [n+1]} \{b_i,b'_i\}$\hfill} ($2n+2$ leaves)
\item \makebox[5.5cm]{$L_3 = \bigcup_{j \in [m]} \{c_j,d_j\}$\hfill} ($2m$ leaves)
\item \makebox[5.5cm]{$L_4 = \{\alpha, \beta, \gamma\}$\hfill} (3 leaves)
\end{itemize}

For each variable $x_i$ in $F$, we create a gadget consisting of two parallel length-$2m{+}2$ paths intersecting at their first and last nodes but otherwise node-disjoint (see Fig.~\ref{fig:var}), where the path taken will determine the variable's truth value. The rooted triples in $R$ corresponding to variable $x_i$'s gadget are:

\begin{itemize}
\item On its {\em positive side}:

$\{b_i b'_i |x_i^1, ~~ b'_i x_i^1 | y_i^1, ~~ x_i^1 y_i^1 | x_i^2, ~~ y_i^1 x_i^2 | y_i^2, ~ ..., ~ x_i^{m-1} y_i^{m-1} | x_i^m, ~~ y_i^{m-1} x_i^m  | y_i^m, ~~ x_i^m y_i^m | b_{i+1}, ~~ y_i^m b_{i+1} | b'_{i+1}\}$
\item On its {\em negative side}:

$\{b_i b'_i | \bar x_i^1, ~~ b'_i \bar x_i^1 | \bar y_i^1, ~~ \bar x_i^1 \bar y_i^1 | \bar x_i^2, ~~ \bar y_i^1 \bar x_i^2 | \bar y_i^2, ~ ..., ~  \bar x_i^{m-1} \bar y_i^{m-1} | \bar x_i^m, ~~ \bar y_i^{m-1} \bar x_i^m | \bar y_i^m, ~~ \bar x_i^m \bar y_i^m | b_{i+1}, ~~ \bar x_i^m b_{i+1} | b'_{i+1}\}$
\end{itemize}

For each clause $C_j = (\tilde x_{i_1}^j \vee \tilde x_{i_2}^j \vee \tilde x_{i_3}^j)$ in $F$, we create a gadget consisting of three (or two, in the case of a two-literal clause) parallel length-3 paths, intersecting in their first and fourth nodes, followed by one additional (shared) edge (see Fig.~\ref{fig:clause}), where the path taken (the {\em witness path}) will correspond to which of $C_j$'s literal satisfies the clause (or one among them, in the case of multiple true literals). The second node of $C_j$'s witness path (of the form $c_j \tilde x_i^j$, and corresponding to the appearance of literal $\tilde x_i$) is its {\em witness node}. The rooted triples in $R$ corresponding to clause $C_j$'s gadget are:

\begin{itemize}
\item $\{c_j d_j | x_i^j, ~~ c_j x_i^j | y_i^j, ~~ c_j y_i^j | c_{j+1}\}$, ~ for each positive appearance of a variable $x_i$ in $C_j$
\item $\{c_j d_j | \bar x_i^j, ~~ c_j \bar x_i^j | \bar y_i^j, ~~ c_j \bar y_i^j | c_{j+1}\}$, ~ for each negative appearance of a variable $x_i$ in $C_j$
\item $c_j c_{j+1} | d_{j+1}$, ~if $j<m$%, ~and~ $c_j c_{j+1} | \gamma$, ~if $j=m$
\end{itemize}

Finally, $R$ has the following triples connecting the pieces together, connecting the source node $\alpha \beta$ to a chained-together series of variable gadgets, the last of which is connected (via an intermediate node) to the first of a chained-together series of clause gadgets, the last of which is connected to the destination node $c_{m+1}\gamma$:

\begin{itemize}
\item $\{\alpha \beta | b_1, ~~ \beta b_1 | b'_1\}$
\item $\{b_{n+1} b'_{n+1} | c_1, ~~ b'_{n+1} c_1 | d_1\}$
\item $c_m c_{m+1} | \gamma$
\end{itemize}

%The source node $ab$ is connected to a chained together series of variable gadgets, the last of which is connected to the first of the chained-together series of clause gadgets, the last of which is connected to the destination node $ac$. In addition, there are certain connections from nodes in clause gadgets to nodes in variable gadgets, specified as follows...

It is important to remember that all these connections are 1-2-hyperarcs. Sometimes both heads will be nodes within variable and clause gadgets, but in most cases one of the two heads will be a sink node whose only role is to permit the hyperarc to conform to the required structure.

% !TEX root = rootedtriples-arx.tex

\section{The Proof}

Clearly \entailprob\ is in NP: if we guess the subset $R' \subset R$, then we can verify both that $R'$ is consistent and that $R' \vdash t$ by executing Aho et al.\ \cite{aho1981inferring}'s polynomial-time BUILD algorithm on $R'$ \cite{bryant1995extension}. \minapdh\:is as well: guess the path, and check that it is acyclic.

Now we prove hardness, arguing that $R$ contains a consistent subset entailing $\alpha\beta|\gamma$ iff $H$ contains an acyclic path $P$ from $\alpha\beta$ to $c_{m+1}\gamma$ iff $F$ admits a satisfying assignment $\mathbf{v}(\cdot)$, in two steps.

%\textbf{3SAT} reduces to \apdh\:(which is therefore NP-Complete):

\subsection{Acyclic Path $\Leftrightarrow$ Satisfying Truth Assignment}

First we argue that acyclic paths correspond to satisfying truth assignments.

\begin{lemma}
There is an an acyclic path $P$ from $\alpha\beta$ to $c_{m+1}\gamma$ iff $F$ admits a satisfying truth assignment $\mathbf{v}(\cdot)$.
%If there is a satisfying truth assignment, then there is an acyclic path from $\alpha\beta$ to $c_{m+1}\gamma$.
\end{lemma}
\begin{proof}
($\Leftarrow$)
We construct $P$ based on the satisfying truth assignment $\mathbf{v}(\cdot)$. For each $x_i$, if $\mathbf{v}(x_i) =$ true, then we have $P$ take the $x_i$ gadget's {\em negative} side (drawn lower in Fig.\ \ref{fig:var}); otherwise, we have $P$ take its {\em positive} side (drawn higher). Then within each $C_j$, among the three possible witness paths, we choose one corresponding to the appearance of a literal that is true under $\mathbf{v}(\cdot)$ (which by assumption must exist).
%, which therefore points only to sides of variable gadgets that $P$ did {\em not} take. 
For the remainder of $P$ we have it take the required connecting arcs.

Notice that any potential cycle contained within $P$ must necessarily involve an arc $a$ whose tail lies within some clause $C_j$'s gadget and one of whose heads lies within some variable $x_i$'s gadget. But the witness path taken within $C_j$'s gadget corresponds to an appearance $\tilde x_i^j$ of a true literal $\tilde x_i$. If it a positive literal, i.e., $\tilde x_i=x_i$, then $\mathbf{v}(x_i) =$ true, and we would have had $P$ take the $x_i$ gadget's {\em negative} side---not the side $a$ is pointing to. Similarly, if $\tilde x_i=\bar x_i$, then $\mathbf{v}(x_i) =$ false, $P$ would have taken the $x_i$ gadget's {\em positive} side, which again is then not the side $a$ is pointing to. Thus $P$ must be acyclic.

($\Rightarrow$)
We read off $\mathbf{v}(\cdot)$ from the choices the path $P$ makes when passing through the series of variable gadgets, i.e., setting $\mathbf{v}(x_i) = $ true if $P$ used the $x_i$ gadget's {\em negative} side, and setting it to false otherwise. Now, consider a clause $C_j$'s gadget, whose witness node points to a node $v'$ within some variable $x_i$'s gadget. Because $P$ is acyclic, we know that $v'$ is on the side of $x_i$'s gadget that $P$ did {\em not} use. Since we chose the truth value $\mathbf{v}(x_i)$ that makes literals within nodes on the $x_i$ gadget's {\em unused} side true, this implies that literal corresponding to $C_j$'s witness node is true, thus satisfying the clause.
\end{proof}

Thus we have proven:

\begin{theorem}\label{thm:apdh}
\apdh\:is NP-Complete.
\end{theorem}

Since an infeasible solution is defined to have infinite cost, an algorithm with {\em any} approximation factor would allow us to distinguish between positive and negative problem instances, which immediately implies:
\begin{corollary}
Approximating \minapdh\:to within {\em any} factor is NP-hard.
\end{corollary}

Even if we restrict ourselves to problem instances admitting feasible solutions, this ``promise problem" \cite{goldreich2006promise} special case is hard to approximate within any reasonable factor.

%More formally, we obtain the following hardness of approximation.

\begin{corollary}
%In the promise problem \cite{goldreich2006promise} special case in which an acyclic B-hyperpath is assumed to exist, approximating \minapdh\:to within factor $|V|^{1-\epsilon}$ for any $\epsilon>0$ is NP-hard.
\minapdhprom\:is NP-hard to approximate to within factor $|V|^{1-\epsilon}$ for all $\epsilon>0$.
\end{corollary}
\begin{proof}
Let $\epsilon \in (0,1/4)$. Let $I = \langle H,u,v \rangle$ be an instance \minapdhprom. For this proof only, let $n=|V(H)|$.
We extend $H$ to a new hypergraph $H'$ by adding polynomially many dummy arcs forming an acyclic $u{\to}v$ path of length $n^{\lceil1+1/\epsilon\rceil}$, yielding a hypergraph $H'$ with $n' = n+n^{\lceil1+1/\epsilon\rceil}-1$ nodes and a new problem instance $I' = \langle H',u,v \rangle$, which is also by construction a valid problem instance of \minapdhprom. If $H$ has a path ($I$ is a positive instance), then $H'$ (like $H$) has one of some length $OPT^+ \le n-1$; if not ($I$ is a negative instance), then $H'$'s only path has length $OPT^- = n^{\lceil1+1/\epsilon\rceil}$.

Now, suppose there were a $|V|^{1-\epsilon}$-approximation algorithm.

Running it on the $I'$ resulting from a negative $I$ will yield a path of length at least $LB^- = OPT^- = n \cdot n^{\lceil1/\epsilon\rceil}$.

Running it on the $I'$ resulting from a positive $I$ will yield a path of length of at most
\begin{align*}
UB^+ &= n'^{1-\epsilon}\cdot (n-1)\\
&= (n+n^{\lceil1+1/\epsilon\rceil}-1)^{1-\epsilon}\cdot (n-1)\\
& < (n-1) \cdot n^{\lceil1+1/\epsilon\rceil\cdot(1-\epsilon)} + (n-1)^{2-\epsilon}\\
&< n \cdot n^{\lceil1+1/\epsilon\rceil\cdot(1-\epsilon)},
\end{align*}
where the last inequality follows from the fact that $(n-1)^{2-\epsilon} < n^{\lceil1+1/\epsilon\rceil\cdot(1-\epsilon)}$ for $\epsilon<1/4$.

Because $\lceil1+1/\epsilon\rceil\cdot(1-\epsilon) = 1+\lceil1/\epsilon\rceil - \epsilon - \epsilon \lceil 1/\epsilon \rceil < \lceil 1/\epsilon \rceil$, it follows that $UB^+ < LB^-$.

%$\lceil1/\epsilon\rceil$ vs $\lceil1+1/\epsilon\rceil\cdot(1-\epsilon)$
%
%$\lceil1/\epsilon\rceil$ vs $\lceil1/\epsilon\rceil\cdot(1-\epsilon) + (1-\epsilon)$
%
%$\lceil1/\epsilon\rceil$ vs $\lceil1/\epsilon\rceil - \lceil1/\epsilon\rceil\cdot\epsilon + (1-\epsilon)$
%
%0 vs $- \lceil1/\epsilon\rceil\cdot\epsilon + (1-\epsilon)$
%
%$\lceil1/\epsilon\rceil\cdot\epsilon > (1-\epsilon)$

Thus by comparing the length $ALG$ of the path in $H'$ found by the hypothetical algorithm to $UB^+$ and $LB^-$, we can decide whether $I$ was positive or negative.
%
%Then any algorithm providing an $f(|V|)$-approximation would allow us to distinguish (for $|V|$ sufficiently large) between instances $H'$ with optimal solution cost $<f(|V|)\cdot|V|$ and those with optimal solution cost $\ge f(|V|) \cdot |V|^2$, thus deciding whether $H$ had a path.
\end{proof}

Second, to extend the reduction to \entailprob, we argue that $H$ is a faithful representation of $R$ in the sense that acyclic paths from $\alpha\beta$ to $c_{m+1}\gamma$ (or acyclic supersets of such paths) correspond to consistent subsets entailing $\alpha\beta|\gamma$, and vice versa.

\subsection{Consistent Entailing Subset $\Leftarrow$ Acyclic Path}

We prove this direction via two lemmas, proving that the set of triples corresponding to an acyclic path are consistent and entail $\alpha\beta|\gamma$, respectively.

\begin{lemma}\label{lem:consistR}
If there is an acyclic path $P \subseteq A$ from $\alpha\beta$ to $c_{m+1}\gamma$, then $R' = \triples(P)$ is consistent.
\end{lemma}
\begin{proof}
To prove $R'$ consistent, we 
%reinterpret the procedure we carried out in the proof of Lemma \ref{lem:isentailed} as implicitly executing 
step through the execution of Aho et al.\ \cite{aho1981inferring}'s BUILD algorithm running on input $R'$, tracking the state of the resulting Ahograph (see Fig.\ \ref{fig:ahograph}) and the tree being constructed (see Fig.\ \ref{fig:pathtree}) as they progress over time.
%Specifically, each application of dyadic inference rule (\ref{eq:2b}), as we step through $P$ in reverse order, can be interpreted as one iteration (or recursive call) of BUILD.
Each iteration can be interpreted as one application of dyadic inference rule (\ref{eq:2b}), as we step through $P$ in reverse order.

%, of stepping through $P$ in reverse, repeatedly applying dyadic inference rule (\ref{eq:2b})
%
%again step through $P$ in reverse, implicitly executing Aho et al.\ \cite{aho1981inferring}'s BUILD algorithm on input $R'$, picturing the state of the resulting Ahograph (see Fig.\ \ref{fig:ahograph}) and the tree being constructed (see Fig.\ \ref{fig:pathtree}) along the way. 

Observe that in Ahograph $[R', L(R')]$ (see Fig.\ \ref{fig:ahograph}), $\gamma$ will be an isolated A-node, since $R$ contains no triples with $\gamma$ on their LHSs,
%Just as dyadic inference rule (\ref{eq:2b}) carried us through to $\alpha\beta | \sigma$, 
%Observe also that by definition of $R'$, the A-nodes $L(R') - \{\gamma\}$ are connected in one component of Ahograph $[R', L(R')]$, 
and that the A-edge $\{c_m, c_{m+1}\}$ in $[R', L(R')]$ due to $P$'s last arc (representing $c_m c_{m+1} | \gamma$) is labeled (only) with $\gamma$. Isolated A-node $\gamma$ is removed for BUILD's second iteration, in which we recurse on the remaining A-nodes. The A-edge $\{c_m, c_{m+1}\}$ will therefore no longer appear in that iteration's Ahograph, $[R', L(R') - \{\gamma\}]$, rendering $c_{m+1}$ an isolated A-node.

We claim that this pattern will continue to obtain the rest of the way back to the start of $P$, with a new A-node becoming isolated  by the start of each iteration, causing an A-edge to be removed, and hence isolating another A-node for the succeeding iteration.
More precisely, number the iterations in reverse, counting downward from $\ell=|P|$, where iteration $\ell$ is the first iteration executed (the one in which A-node $o_\ell=\gamma$ is observed isolated and removed for the second iteration, thus also deleting the A-edge $\{p_\ell,q_\ell\}=\{c_m,c_{m+1}\}$).

Let $H_P$ be the subhypergraph of $H$ induced by the nodes $\{\text{t}(e) \cup \text{h}(e) : {e \in P}\}$, and recall that nodes in $H$ correspond to unordered pairs of A-nodes, and thus to potential A-edges in the Ahograph.

Because $P$ is a path (and because there are no arcs between two witness paths or between the two sides of a variable gadget), each non-sink node in $H_P$ will have out-degree 1 ($c_{m+1}\gamma$ has out-degree 0). Therefore $[R', L(R')]$ will have exactly $\ell$ A-edges, each with one label.
%
%%of $H$ induced by $P$ has out-degree at most one. Therefore no A-edge in the Ahograph has more than one label. Note that $\alpha$ and $\beta$ will not be labels of any A-edge.
%
Because $P$ is acyclic, each node in $H_P$ except $\alpha\beta$ (which has in-degree 0) will have in-degree 1. Therefore no label will appear on multiple A-edges in $[R', L(R')]$.
Thus there is a bijection between Ahograph $[R', L(R')]$'s $\ell$ A-edges and the labels appearing on them.

Node $\alpha\beta$ contains two leaves, and each arc in $P$ introduces one additional leaf, for a total of $\ell+2$ leaves in $L(R')$, and hence $\ell+2$ A-nodes in the Ahograph.

Note also that the first arc in $P$ adds label $b_1$ to A-edge $\{\alpha,\beta\}$ in the Ahograph, and that each subsequent arc adds a label to an A-edge incident to the preceding A-edge, and so all $\ell$ of the Ahograph's A-edges lie within a single component. Since isolated A-node $\gamma$ is one of the $\ell+2$ A-nodes, this means that the other component has exactly $\ell+1$ A-nodes, and is therefore a tree.

%To see this, suppose $c_j \tilde x_i^j{\to}\{\tilde x_i^j \tilde y_i^j, c_j \tilde y_i^j\}$ is a back-arc in $P$.
%
%there are no other A-edges in the Ahograph beyond those listed here.
%%
%Note also that no label appears on more than one A-edge. Therefore there is a one-to-one correspondence between the Ahograph's $\ell$ A-edge labels and $P$'s arcs. Every A-node appears (once) as a label, except for $\alpha$ and $\beta$. Thus the Ahograph has $\ell+2$ A-nodes, and its second component is a tree.

Now, we prove by induction that at the start of each iteration $k$ from $\ell$ down to $1$, the Ahograph consists of two components: some isolated A-node $o_k$ and a tree. We have already verified the base case: $o_\ell$ is isolated at the start, in iteration $\ell$, and the $\ell+1$ remaining nodes form a tree. Assume the claim is true for each iteration from $\ell$ down through some iteration $k$.

By inspection of $H$, observe that $P$ can be decomposed into a sequence of subpaths:
\begin{align*}
P &~=~ (a_1,a_2) \circ P_{x_1} \circ P_{x_2} \circ \cdots \circ P_{x_n} \circ  (a_3,a_4) \circ P_{C_1} \circ P_{C_2} \circ \cdots \circ P_{C_m},
\end{align*}
where $P_{x_i}$ is one of the two length-$2m{+}2$ subpaths passing through $x_i$'s gadget (from $b_ib'_i$ to $b_{i+1}b'_{i+1}$), $P_{C_j}$ is one of the three length-4 subpaths traversing one of the $C_j$ gadget's three witness paths (from $c_jd_j$ to $c_jc_{j+1}$) {\em plus the succeeding arc with tail $c_jc_{j+1}$}, and $a_1,...,a_4$ are shorthand names for the other four connecting arcs.

For each $C_j$, let $\hat i_j$ denote the index $i$ of the variable $x_i$ appearing in the witness path of $C_j$ (through literal appearance $\tilde x_{\hat i_j}^j$) that $P$ traverses.

Now, consider the {\em reverse linear ordering} shown in Table \ref{tbl:linord} of the A-node-labeled A-edges of Ahograph $[R', L(R')]$, with each subsequence of A-node-labeled A-edges tagged to indicate the corresponding subpath of $P$ (compare to Figs.\ \ref{fig:ahograph} and \ref{fig:overview}, and in particular compare each ($P_{C_j}$) to Fig.\ \ref{fig:clause} and each ($P_{x_i}$) to Fig.\ \ref{fig:var}). %, for an acyclic path $P$: 

\begin{table}[t]
\caption{Reverse linear ordering of the A-edges corresponding to arcs in path $P$ from $\alpha\beta$ to $c_{m+1}\gamma$.}
\vskip -.45cm
\centering
\begin{alignat}{7}
(~&\{c_m,c_{m+1}\}:\gamma, ~~&&\{c_m,\tilde y_{\hat i_m}^m\}:c_{m+1}, ~~&&\{c_m,\tilde x_{\hat i_m}^m\}:\tilde y_{\hat i_m}^m, ~~&&\{c_m,d_m\}:\tilde x_{\hat i_m}^m,\label{eq:clausem}\tag{$P_{C_m}$}\\
&\{c_{m-1},c_m\}:d_m, ~~&&\{c_{m-1},\tilde y_{\hat i_{m-1}}^{m-1}\}:c_m, ~~&&\{c_{m-1},\tilde x_{\hat i_{m-1}}^{m-1}\}:\tilde y_{\hat i_{m-1}}^{m-1}, ~~&&\{c_{m-1},d_{m-1}\}:\tilde x_{\hat i_{m-1}}^{m-1},\label{eq:clausemm1}\tag{$P_{C_{m-1}}$}\\
& ...,\nonumber\\
&\{c_1,c_2\}:d_2, ~~&&\{c_1,\tilde y_{\hat i_1}^1\}:c_2, ~~&&\{c_1,\tilde x_{\hat i_1}^1\}:\tilde y_{\hat i_1}^1, ~~&&\{c_1,d_1\}:\tilde x_{\hat i_1}^1,\label{eq:clause1}\tag{$P_{C_1}$}\\
&\{b'_{n+1},c_1\}:d_1, ~~&&\{b_{n+1},b'_{n+1}\}:c_1,\label{eq:var2clause}\tag{$a_3,a_4$}\\
&\{\tilde y_n^m,b_{n+1}\}:b'_{n+1}, ~~&&\{\tilde x_n^m,\tilde y_n^m\}:b_{n+1}, ~~&&\{\tilde y_n^{m-1},\tilde x_n^m\}:\tilde y_n^m, ~ \quad..., ~~&&\{\tilde x_{n}^1,\tilde y_{n}^1\}:\tilde x_{n}^2,\nonumber\\
&\quad\{b'_n,\tilde x_n^1\}:\tilde y_n^1, ~~&&\{b_n,b'_n\}:\tilde x_n^1,\label{eq:varn}\tag{$P_{x_n}$}\\
&\{\tilde y_{n-1}^m,b_n\}:b'_n, ~~&&\{\tilde x_{n-1}^{m}:b_n,\tilde y_{n-1}^{m}\}, ~~&&\{\tilde y_{n-1}^{m-1},\tilde x_{n-1}^m\}:\tilde y_{n-1}^m,~~ ..., ~~&&\{\tilde x_{n-1}^1,\tilde y_{n-1}^1\}:\tilde x_{n-1}^2,\nonumber\\
&\quad\{b'_{n-1},\tilde x_{n-1}^1\}:\tilde y_{n-1}^1, ~~&&\{b_{n-1},b'_{n-1}\}:\tilde x_{n-1}^1,\label{eq:varnm1}\tag{$P_{x_{n-1}}$}\\
&...,\nonumber\\
&\{\tilde y_{1}^m,b_{2}\}:b'_{2}, ~~&&\{\tilde x_{1}^m,\tilde y_{1}^m\}:b_{2}, ~~&&\{\tilde y_1^{m-1}:\tilde y_1^m,\tilde x_1^m\},~ \quad..., ~~&&\{\tilde x_1^1,\tilde y_1^1\}:\tilde x_1^2,\nonumber\\
&\{b'_1,\tilde x_1^1\}:\tilde y_1^1, ~~&&\{b_1,b'_1\}:\tilde x_1^1,\label{eq:var1}\tag{$P_{x_1}$}\\
&\{\beta,b_1\}:b'_1, ~~&&\{\alpha,\beta\}:b_1~)\tag{$a_3,a_4$}
\end{alignat}
\vskip -.25cm
\label{tbl:linord}
\end{table}

Observe the relationship between each {\em A-edge$:$label} entry in the ordering and its predecessor: for each entry $k-1$ for $k \le \ell$, its label $o_{k-1}$ is one of the two A-nodes occurring in entry $k$'s (i.e., the preceding entry's) A-edge, say $p_k$, and crucially, $p_k$ has no further appearances in the list past entry $k-1$. That is, {\em at the start of iteration $k$, $p_k$ is a leaf in the Ahograph}.

Therefore iteration $k$'s removal of the then-isolated A-node $o_k$ and of the (only) label from A-edge $\{p_k,p_k'\}$ disconnects the component that had contained that A-edge in two, thus isolating $p_k = o_{k-1}$ in time for iteration $k-1$. 
%Thus A-edge $\{p_k,p_k'\}$'s removal in iteration $k$ means that $p_k = o_{k-1}$ will be an isolated A-node at the start of iteration $k-1$. 
The remainder of that same component is of course also a tree, thus proving the inductive claim for iteration $k-1$.

%Note also that one of A-edge $\{p_k,p_k'\}$'s two endpoints (say, $p_k$) the {\em label} of the succeeding entry $o_{k-1}{:}\{p_{k-1},p_{k-1}'\}$ (i.e., $p_k=o_{k-1}$) and does not appear as an {\em endpoint} in any subsequent A-edge in the ordering. That is, at the start of the iteration $k$, $p_k$ is a leaf in the Ahograph. Therefore the removal of A-edge $\{p_k,p_k'\}$ means that $p_k = o_{k-1}$ is an isolated A-node at the start of iteration $k-1$. The remainder of that component within the Ahograph is of course still a tree, thus proving the inductive claim.

Since the BUILD algorithm running on $R'$ therefore never encounters a connected Ahograph, it follows \cite{bryant1995extension} that $R'$ is consistent.
%
%%Therefore each removal of the next {\em label$:$edge} entry causes the isolation of the A-node appearing as the subsequent entry's label.
%
%%The initial Ahograph consists of two components, one of which is an isolated A-node ($\gamma$). It has $\ell$ A-edges, corresponding to $P$'s $\ell$ arcs. Every A-node appears (once) as a label, except for $\alpha$ and $\beta$. Thus the Ahograph has $\ell+2$ A-nodes, and its second component is a tree. Therefore each deletion of an A-edge from that second component increases the number of components by one.
%
%Finally, observe that in each iteration, the A-edge deleted is incident to a leaf, thus rendering it an isolated A-node.
%
%
%\textbf{HERE}
%
%
%If $o_k$ nonetheless is still not isolated, that means that triple $p_kq_k | o_k$ was not removed in the previous iteration. 
%
%in each iteration a new A-node $u_k$ becomes isolated
%
%encountering a new isolated A-node in each iteration, 
%
%
%as the process implicitly constructs a caterpillar tree (see Fig.\ \ref{fig:caterpillar}), thus implying that $R'$ is consistent.
\end{proof}

\begin{lemma}\label{lem:isentailed}
%Let $A' \subseteq A$. 
If there is an acyclic path $P \subseteq A$ from $\alpha\beta$ to $c_{m+1}\gamma$, then $R' = \triples(P)$ entails $\alpha\beta|\gamma$.
%there exists a subset $A' \subseteq A$ for which $R' = \triples(A')$ is a consistent subset of $R$ entailing $\alpha\beta|\gamma$.
\end{lemma}
\begin{proof}
Let $\ell=|P|$, and for each $k \in [\ell]$, let $a_k = u_k{\to}\{v_k,v_k'\}$ denote the $k$th arc in $P$, with $u_k = p_kq_k$ and $\{v_k,v_k'\} = \{p_ko_k, q_ko_k \}$ (for some leaves $p_k,q_k,o_k$), representing triple $p_k q_k | o_k = u_k | o_k$. The names $v_\ell,v_\ell'$ can be assigned so that $v_\ell = c_{m+1}\gamma$ (and $v_\ell'$ is $a_\ell$'s non-destination sink node $c_m \gamma$), but we can also simply represent $a_\ell$ as $u_\ell|o_\ell = c_mc_{m+1} | \gamma$.
%so hyperarc $a_\ell = u_\ell{\to}\{v_\ell, v_\ell'\}$, representing triple $u_\ell|o_\ell = p_\ellq_\ell | o_\ell = c_mc_{m+1} | \gamma$. 
Similarly, for each $k \in [2,...,\ell]$, we can assign the names $v_{k-1},v_{k-1}'$ so that $v_{k-1}=u_k$, i.e., $\text{t}(a_k) = u_k = v_{k-1} \in \text{h}(a_{k-1})$,
%and so hyperarc $a_{k-1} = u_{k-1}{\to}\{u_k, v_{k-1}'\}$, representing triple $p_{k-1}q_{k-1} | o_{k-1} = u_{k-1}|o_{k-1}$.
and we can also simply represent $a_{k-1}$ as $u_{k-1}|o_{k-1}$.

Now consider arcs $a_{\ell-1}$ and $a_\ell$, and observe that their representations as $u_{\ell-1}|o_{\ell-1}$ and $u_\ell|o_\ell$ (respectively) are 
%$a_{\ell-1} = u_{\ell-1}{\to}\{u_\ell, v_{\ell-1}'\}$ and $a_\ell = u_\ell{\to}\{v_\ell, v_\ell'\}$, representing triples $u_{\ell-1}|o_{\ell-1}$ and $u_\ell|o_\ell$, respectively, and observe that this pair of triples is 
of exactly the form that dyadic inference rule (\ref{eq:2b}) applies to: via that inference rule, we can derive the triple $u_{\ell-1}|o_\ell$, which have the form $c_m \tilde y_i^m | \gamma$ (for some $i$). % or $c_j \bar y_i^m | \gamma$.

%, where the names $v_k,v_k'$ are assigned so that $v_k$ is $u_k$'s successor on $P$, i.e., $v_k = u_{k+1}$ (except when $k=\ell$, in which case $v_k=c_{m+1}\gamma$), and therefore, $u_k=v_{k-1}$, i.e., $u_k = \text{t}(a_k) \in \text{h}(a_{k-1}) = \{v_{k-1},v_{k-1}'\}$
%
%That is, $\text{t}(a_k) = p_kq_k$, and $o_{k-1} \in \{p_k, q_k\}$ and $|\{p_k, q_k\} \cap \{p_{k-1}, q_{k-1}\}|=1$, and so $a_k$ represents $p_k q_k | o_k$. Then observe that by applying the second dyadic inference rule (recall Def.\ \ref{def:2ord}) to the triples represented by the path's last two arcs (i.e., $p_{\ell-1} q_{\ell-1} | o_{\ell-1}$ and $p_\ell q_\ell | o_\ell$), we can infer $p_{\ell-1} q_{\ell-1} | o_\ell$. 

Similarly, by applying this inference rule to the triple represented $P$'s third-to-last arc (i.e., $u_{\ell-2} | o_{\ell-2}$) and the derived triple  $u_{\ell-1} | o_\ell$, we can derive $u_{\ell-2} | o_\ell$. Via a total of $\ell-1$ such applications of dyadic inference rule (\ref{eq:2b}) (each corresponding to one BUILD iteration), we can derive $u_1 | o_\ell = p_1 q_1 | o_\ell$ = $\alpha \beta | \gamma$.
%
%Second, by Lemma \ref{lem:consistR} proved below, $R'$ is indeed consistent.
\end{proof}

Thus we have:
\begin{corollary}[$\Leftarrow$]\label{cor:consentails}
%Let $A' \subseteq A$. 
If there is an acyclic path $P \subseteq A$ from $\alpha\beta$ to $c_{m+1}\gamma$, then $R' = \triples(P)$ is consistent and entails $\alpha\beta|\gamma$.
%there exists a subset $A' \subseteq A$ for which $R' = \triples(A')$ is a consistent subset of $R$ entailing $\alpha\beta|\gamma$.
\end{corollary}

%mpj
%\newpage

%mpj
\begin{figure}%[t!]
\centering
	\begin{subfigure}{\textwidth}
\centering
%\input{triples-fig-var-leaves}
%\captionsetup{justification=centering}
%\includegraphics[width=\linewidth, clip=true, trim={0cm 10cm 3cm 5cm}]{triplescaterpillar.jpg}
%\input{triples-fig-tree}
\makebox[1\textwidth][c]{ %\hskip -2cm
\includegraphics[width=1\textwidth]{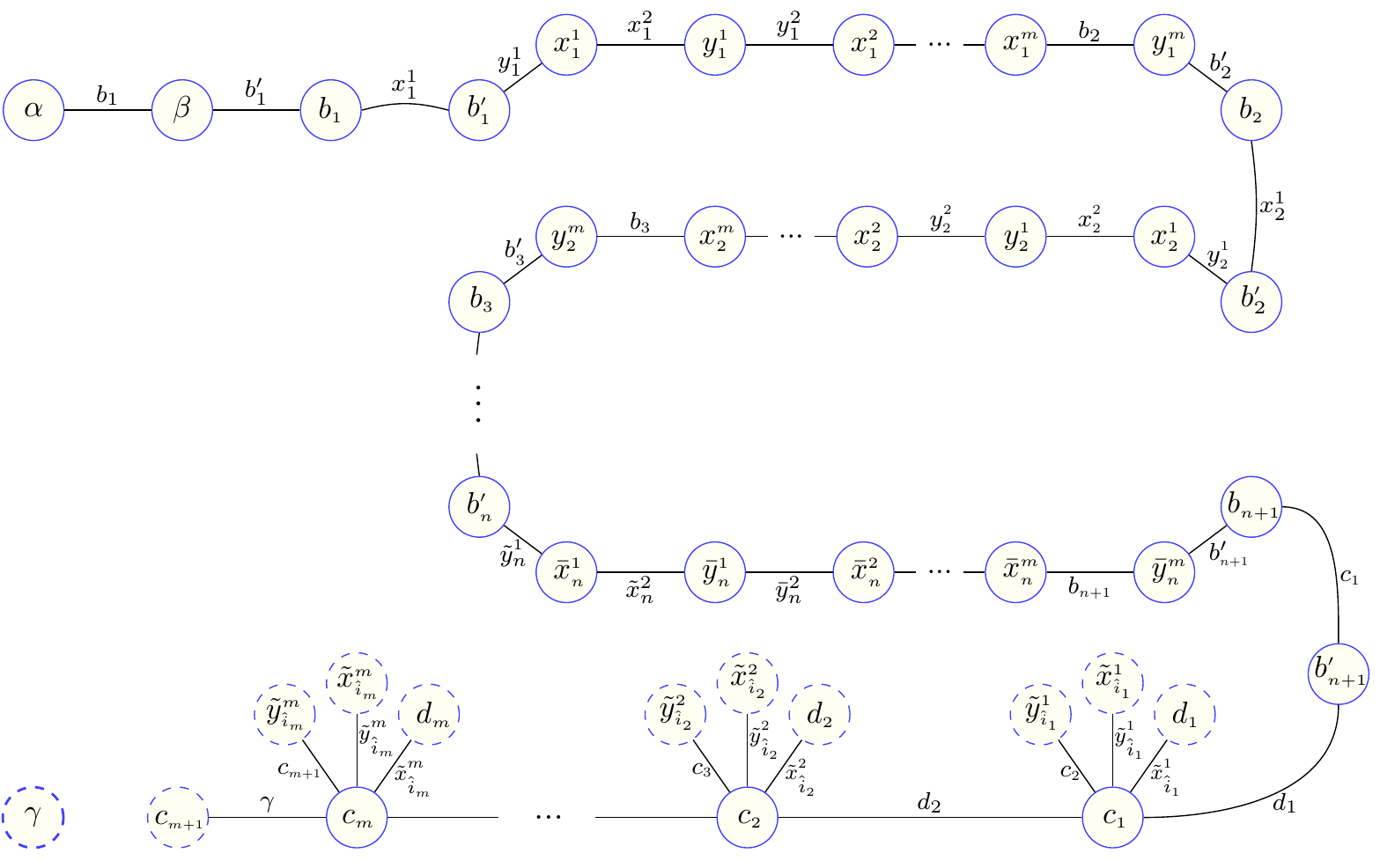}
}
%\vskip .05cm
%\vskip -.05cm
\caption{The initial Ahograph $[R',L(R')]$ constructed based on an underlying {\em acyclic} path $P$. Because $P$ is acyclic, the node  $\tilde x_{\hat i_j}^j \tilde y_{\hat i_j}^j$ pointed to by each $C_j$'s witness node $c_j \tilde x_{\hat i_j}^j$ is a {\em sink node in $H$}, and thus A-nodes $\tilde y_i^j$ and $\tilde x_i^j$ are {\em leaves in the Ahograph}. As a result, at each BUILD iteration, the Ahograph has two components, a tree and an isolated node.
%, distinct from the portion of the Ahograph corresponding to $P$'s traversal of the variable gadgets.% (as occurs in Fig.\ \ref{fig:ahocycle}).
%
%Because $P$ is acyclic, 
%, where $P$ is an acyclic path in the hypergraph from $\alpha\beta$ to $c_{m+1}\gamma$,
%%(specifically, the red path shown in Fig.\ \ref{fig:overview}), 
%and $R'=\triples(P)$. 
%%Here $x_2$ is false and $x_1,x_3,x_4$ are true, and the witness literals for clauses $C_1,C_2,C_3,C_4$ are $\bar x_n, \bar x_1, \bar x_n, \bar x_1$, respectively. 
%%If $P$ had contained a cycle involving some clause $C_j$, then $C_j$'s $x$ and $p$ nodes in the resulting Ahograph would have 
%Let $\tilde x_i^j \tilde y_i^j$ be the variable gadget node pointed to by the $C_j$ clause's witness node. Then, because $P$ is acyclic, A-nodes $\tilde x_i^j$ and $\tilde y_i^j$ are both leaves {\em in the Ahograph}, rather than appearing in the portion of the Ahograph corresponding to $P$'s traversal of the variable gadgets (as occurs in Fig.\ \ref{fig:ahocycle}).
\label{fig:ahograph}
}
%\end{figure}
%\begin{figure}[t!]
\end{subfigure}
\vskip .35cm
\centering
	\begin{subfigure}{\textwidth}
\centering
%\captionsetup{justification=centering}
%\includegraphics[width=\linewidth, clip=true, trim={0cm 10cm 3cm 5cm}]{triplescaterpillar.jpg}
%\input{triples-fig-tree}
\includegraphics[width=.75\textwidth]{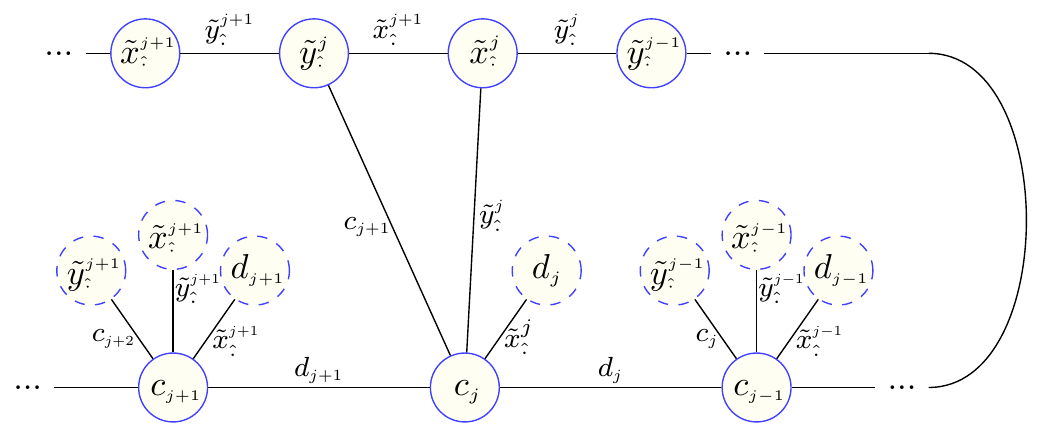}
%\vskip .05cm
%\vskip -.05cm
\caption{A fragment of the Ahograph resulting from the cause of a cycle in $P$, viz., a back-arc in which some clause $C_j$'s witness node $c_j \tilde x_{\hat i_j}^j$ points to a {\em non-sink} node $\tilde x_{\hat i_j}^j \tilde y_{\hat i_j}^j$ that lies earlier in the path $P$ (specifically, within $x_{\hat i_j}$'s gadget), and thus A-nodes $\tilde y_{\hat i_j}^j$ and $\tilde x_{\hat i_j}^j$ are {\em non-leaves} in the Ahograph. This yields the triangle $c_j-\tilde y_{\hat i_j}^j - \tilde x_{\hat i_j}^j$. After the BUILD iteration in which A-node $c_{j+1}$ and A-edge $\{c_j, \tilde y_{\hat i_j}^j\}$ are removed, $\tilde y_{\hat i_j}^j$ will still {\em not} be isolated, and hence the Ahograph will at that point consist of a single component, indicating $R'$ is inconsistent.
%Ahograph $[R',L(R')]$, where $P$ is an acyclic path in the hypergraph from $\alpha\beta$ to $c_{m+1}\gamma$ (specifically, the red path shown in Fig.\ \ref{fig:overview}), and $R'=\triples(\arcs(P))$. Here $x_2$ is false and $x_1,x_3,x_4$ are true, and the witness literals are $\tilde x_n, \tilde x_1, \tilde x_n, \tilde x_1$ for clauses $C_1,C_2,C_3,C_4$, respectively. 
%%If $P$ had contained a cycle involving some clause $C_j$, then $C_j$'s $x$ and $p$ A-nodes in the resulting Ahograph would have  
%Because $P$ is acyclic, each clause $C_j$'s $x$ and $p$ A-nodes are leaves {\em in the Ahograph}, rather than being A-nodes in the portion of the Ahograph corresponding to $P$'s traversal of the variable gadgets. 
\label{fig:ahocycle}
}
\end{subfigure}
\caption{The Ahograph $[R',L(R')]$ is constructed when Aho et al.'s BUILD algorithm is executed on the triples $R' = \triples(P)$ corresponding to the path $P$ from $\alpha\beta$ to $c_{m+1}\gamma$. The existence of any cycle within $P$ involving some clause $C_j$'s witness node $c_j \tilde x_{\hat i_j}^j$ pointing to $\tilde x_{\hat i_j}^j \tilde y_{\hat i_j}^j$ (i.e., arc $c_j \tilde x_{\hat i_j}^j \to \{c_j \tilde y_{\hat i_j}^j, \tilde x_{\hat i_j}^j \tilde y_{\hat i_j}^j\}$) will be reflected in the locations in the Ahograph of A-nodes $\tilde y_{\hat i_j}^j$ and $\tilde x_{\hat i_j}^j$.}
\end{figure}

\begin{figure}
\centering
\includegraphics[width=1\textwidth]{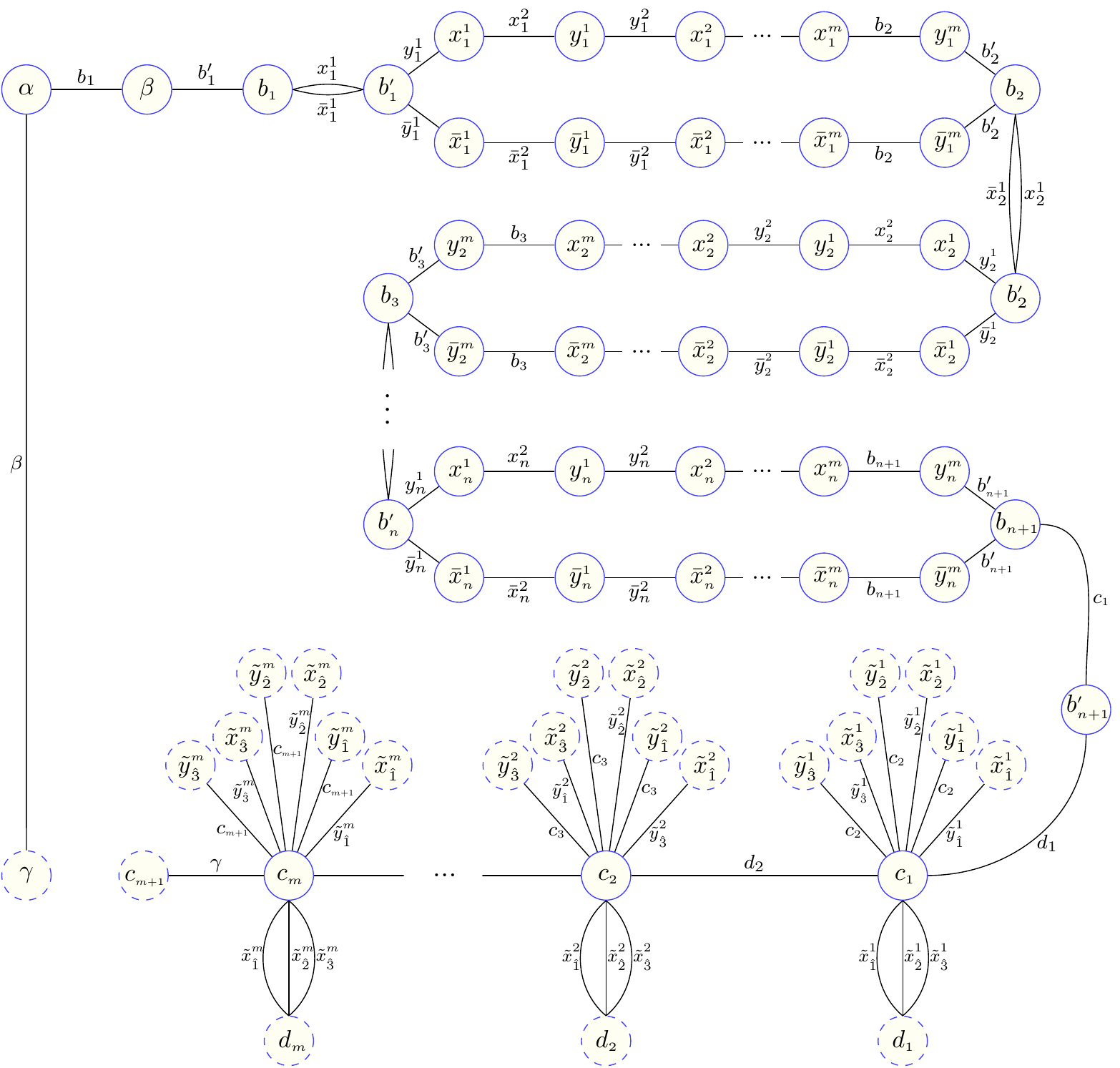}
\vskip -,15cm
\caption{The Ahograph $[R',L(R')]$ for $R'=\triples(A')$, where $A'$ is a maximal acyclic subset of $A$ and the three variable gadgets shown (for $x_1,x_2,x_n$) are each maximal in the sense that both sides are included and the three clause gadgets shown (for $C_1,C_2,C_m$) are each maximal in the sense that all three witness paths are included. (Hence all the appearances within clauses $C_1,C_2,C_m$ must be of variables $x_3,...,x_{n-1}$.) If there is no disconnection between either $b_2$ and $b'_n$ or $c_2$ and $c_m$, then the Ahograph is connected and $R'$ is inconsistent.\label{fig:ahoincompgr}
}
%none of $x_1,x_2,x_n$ appear in any of $C_1,C_2,C_m$.)}
%\end{figure}
%
%
\vskip .35cm
%\begin{figure}[t!]
\centering
\includegraphics[width=.95\textwidth]{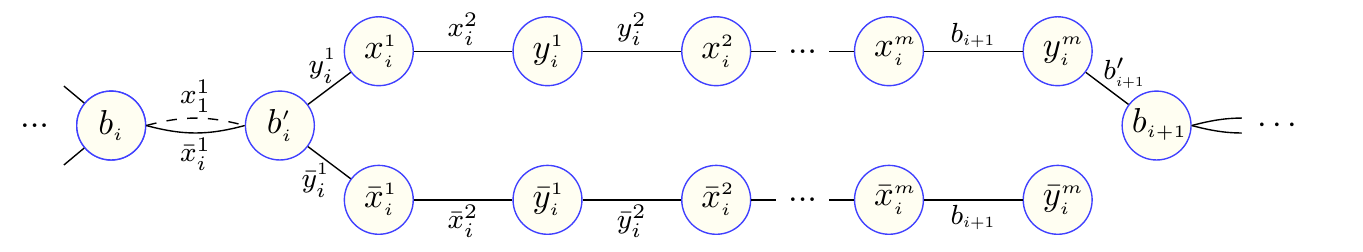}
\vskip -,25cm
\caption{A fragment of the Ahograph corresponding to a variable $x_i$'s gadget (see Fig.\ \ref{fig:var}) if its only two arcs missing from $A'$ are $\arcs(b_ib'_i|x_i^1, ~\bar y_i^m|b_{i+1}|b'_{i+1})$, with the A-edge we consider hypothetically adding in the proof of Lemma \ref{lem:degeneratecutpath} shown dashed.\label{fig:ahodiscordvar}
}
\end{figure}

\subsection{Consistent Entailing Subset $\Rightarrow$ Acyclic Path}

Now we argue for the reverse direction, % (or rather, its contrapositive), %, that satisfying truth assignments correspond to acyclic paths, 
proving through a series of lemmas that if there is no acyclic $\alpha\beta{-}c_{m+1}\gamma$ path, then there will be no consistent triple subset entailed $\alpha\beta|\gamma$.
%via a proof that will depend on two subsequent lemmas proved below.

\begin{lemma}\label{lem:cycpath}
%If $R'$ is a consistent subset of $R$ entailing $\alpha\beta|\gamma$, then $A' = \arcs(R')$ contains the arcs of an acyclic path $P$ from $\alpha\beta$ to $c_{m+1}\gamma$.
%
%Let $A' \subseteq A$. not
Let $A' \subseteq A$. Suppose there exists a {\em cyclic} path $P \subseteq A'$ from $\alpha\beta$ to $c_{m+1}\gamma$. Then $R' = \triples(A')$ is inconsistent.
%Suppose there is no acyclic path $P \subseteq A$ from $\alpha\beta$ to $c_{m+1}\gamma$. Then for every subset $A' \subseteq A$, $R' = \triples(A')$ is  inconsistent or does not entail $\alpha\beta|\gamma$.
%If $A' \subseteq A$ is {\em not} an acyclic arc set subsuming the set $P$ of arcs appearing in some path from $\alpha\beta$ to $c_{m+1}\gamma$, then $R' = \triples(A')$ is not a consistent subset of $R$ entailing $\alpha\beta|\gamma$.
\end{lemma}
\begin{proof}
Intuitively, it can be seen that any cycle in $H$ corresponds to a vicious (inconsistent) cycle of rooted triples. In particular (see Fig.\ \ref{fig:ahocycle}), any clause $C_j$'s component whose witness node points to a variable gadget node already visited on the path (causing a cycle in $H$) will induce a triangle {\em in the Ahograph}, resulting in a single-component Ahograph in the BUILD iteration following the deletion of isolated node $c_{j+1}$ (if no single-component Ahograph has yet been encountered in an earlier iteration), since at that point A-node $\tilde y_{\hat i_j}^j$ will have failed to become isolated.
\end{proof}

Most of the remainder of this subsection will be dedicated to showing {\em constructively} that if $A'$ contains no path from $\alpha\beta$ to $c_{m+1}\gamma$ at all, cyclic or otherwise, then $R'$ does not entail $\alpha\beta|\gamma$. We do so by showing that in the case of such a (consistent) $R'$, there exist trees displaying $R' \cup \{\alpha\beta|\gamma\}$. Therefore assume w.l.o.g.\ that $R'$ is consistent and maximal in the sense that adding any other triple of $R$ to it would either make $R'$ inconsistent or would introduce an $\alpha\beta{-}c_{m+1}\gamma$ path in $A'=\arcs(R')$.
%The remainder of the proof in this subsection will be to show (constructively) that if there exists a consistent $R'$ entailing $\alpha\beta|\gamma$ then there exists a path from $\alpha\beta$ to $c_{m+1}\gamma$.
%$A'$ subsumes no path from $\alpha\beta$ to $c_{m+1}\gamma$ at all, cyclic or otherwise, then $R'$ will not entail $\alpha\beta|\gamma$. Therefore assume w.l.o.g.\ that $R'$ is consistent.
%
%We begin by defining some notation.

Observe that the missing arcs $A^\times = A - A'$ can be thought of as the (source side to sink side) cross arcs of a cut separating source $\alpha\beta$ and sink $c_{m+1}\gamma$.
In the following argument we will refer to hypergraph $H^\gamma = (V \cup \{ \gamma\alpha\}, A' \cup \arc(\gamma\alpha|\beta))$ and its corresponding Ahograph $G^\gamma$.

\begin{figure}
\centering
\makebox[1\textwidth][c]{ %\hskip -3.75cm
%    \captionsetup{width=.9\linewidth}%
   \centering
\ifwabi
   \begin{subfigure}[t]{.14\linewidth}
\clipbox{6.15cm .15cm 5.65cm .5cm}{      \includegraphics[height=.535\textheight,angle=33]{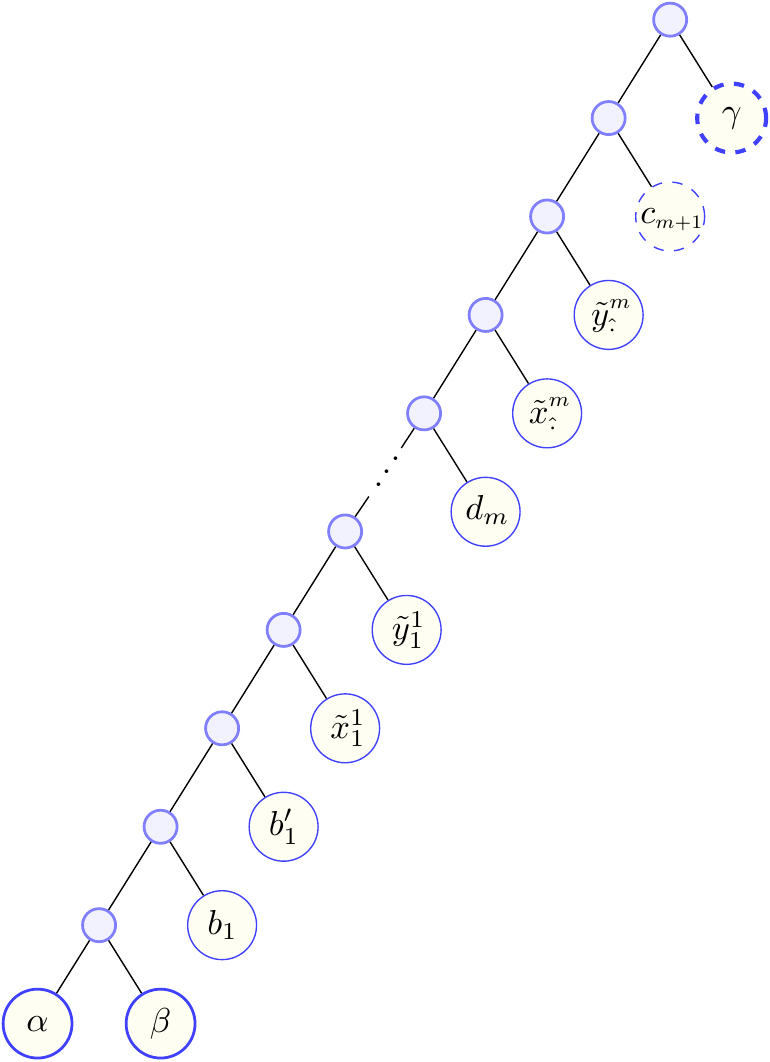}  }
\else
   \begin{subfigure}[t]{.12\linewidth}
\clipbox{6.25cm .15cm 5.75cm .5cm}{      \includegraphics[height=.535\textheight,angle=31.75]{rt-fig-pathtree-arx.pdf}  }
\fi
\caption{$A'$ is an acyclic path from $\alpha\beta$ to $c_{m+1}\gamma$.
%The caterpillar tree determined by the triples corresponding to $P$, which thence satisfies $\alpha\beta|\gamma$.
\label{fig:pathtree}
}
  \end{subfigure}
\ifwabi
\hskip .3cm
  \begin{subfigure}[t]{.26\linewidth}
\else
\hskip .35cm
  \begin{subfigure}[t]{.28\linewidth}
\fi
    \centering
    %\captionsetup{width=1.1\linewidth}%
\ifwabi
\clipbox{2.5cm .15cm 1.5cm 1.4cm}{      \includegraphics[height=.21\textheight,angle=35.6]{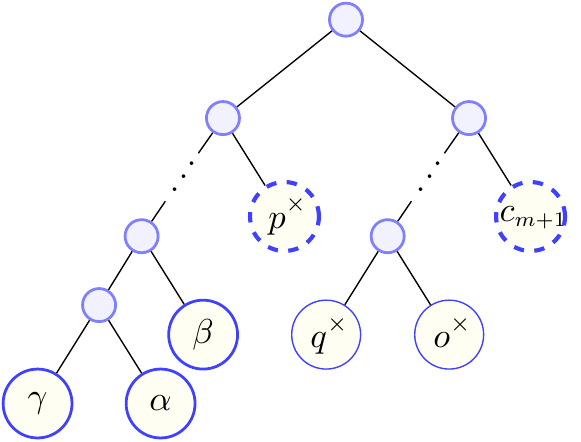} }
\else
\clipbox{2.6cm .15cm 1.6cm 1.45cm}{      \includegraphics[height=.22\textheight,angle=33.6]{rt-fig-cuttreeforced-arx.pdf} }
\fi
\caption{$A^\times$ contains some forced arc $a_\times$.
%A tree consistent with $A'$, i.e., displaying all of $R'=\triples(A')$, but {\em not} displaying $\alpha \beta | \gamma$.
%The tree that can be constructed displaying all triples in $R'=\triples(A')$ but {\em not} $\alpha \beta | \gamma$, if $A'$ is acyclic and there is no path $P \subseteq A'$ from $\alpha\beta$ to $c_{m+1}\gamma$. 
%The root's left subtree contains all leaves appearing in A-nodes before the ``cut'' (plus $\gamma$), and its right subtree contains all leaves appearing after the ``cut'' (except $\gamma$).
\label{fig:cuttreeforced}
}
  \end{subfigure}
\ifwabi
\hskip .35cm
\else
\hskip .35cm
\fi  \begin{subfigure}[t]{.28\linewidth}
    \centering
    %\captionsetup{width=1.1\linewidth}%
\ifwabi
\clipbox{5.45cm .15cm 4.45cm 1.4cm}{      \includegraphics[height=.455\textheight,angle=34.6]{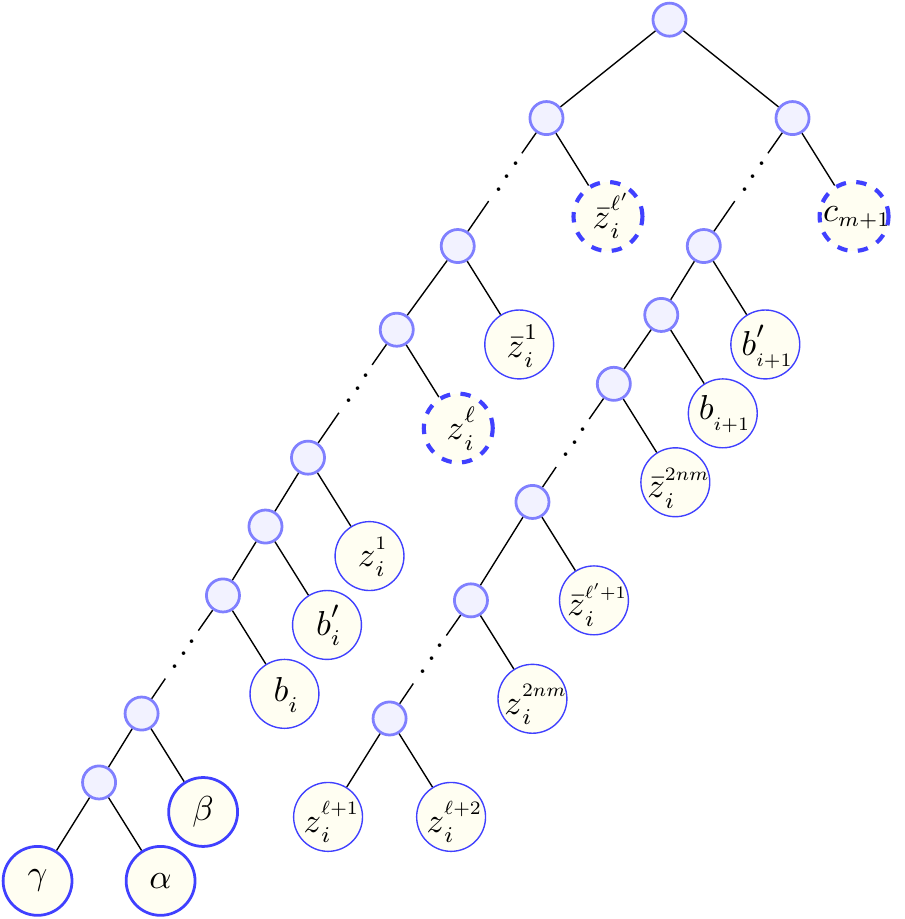} }
\else
\clipbox{5.75cm .15cm 4.7cm 1.25cm}{      \includegraphics[height=.47\textheight,angle=33.6]{rt-fig-cuttreevar-arx.pdf} }
\fi
\caption{$A^\times$ contains $\arcs($ $z_i^\ell z_i^{\ell+1}|z_i^{\ell+2}, ~ \bar z_i^{ell'} \bar z_i^{\ell'+1}|\bar z_i^{\ell'+2})$, where either both or neither of $z_i^\ell,\bar z_i^{\ell'}$ equal $b_i$.
%A tree consistent with $A'$, i.e., displaying all of $R'=\triples(A')$, but {\em not} displaying $\alpha \beta | \gamma$.
%The tree that can be constructed displaying all triples in $R'=\triples(A')$ but {\em not} $\alpha \beta | \gamma$, if $A'$ is acyclic and there is no path $P \subseteq A'$ from $\alpha\beta$ to $c_{m+1}\gamma$. 
%The root's left subtree contains all leaves appearing in A-nodes before the ``cut'' (plus $\gamma$), and its right subtree contains all leaves appearing after the ``cut'' (except $\gamma$).
\label{fig:cuttreevar}
}
  \end{subfigure}
\ifwabi
\hskip .45cm
  \begin{subfigure}[t]{.3\linewidth}
\else
\hskip .35cm
  \begin{subfigure}[t]{.28\linewidth}
\fi
    \centering
    %\captionsetup{width=1.1\linewidth}%
%\hskip -.35cm
\ifwabi
\clipbox{7.05cm 1.05cm 5.85cm .55cm}{   \includegraphics[height=.7\textheight,angle=32.75]{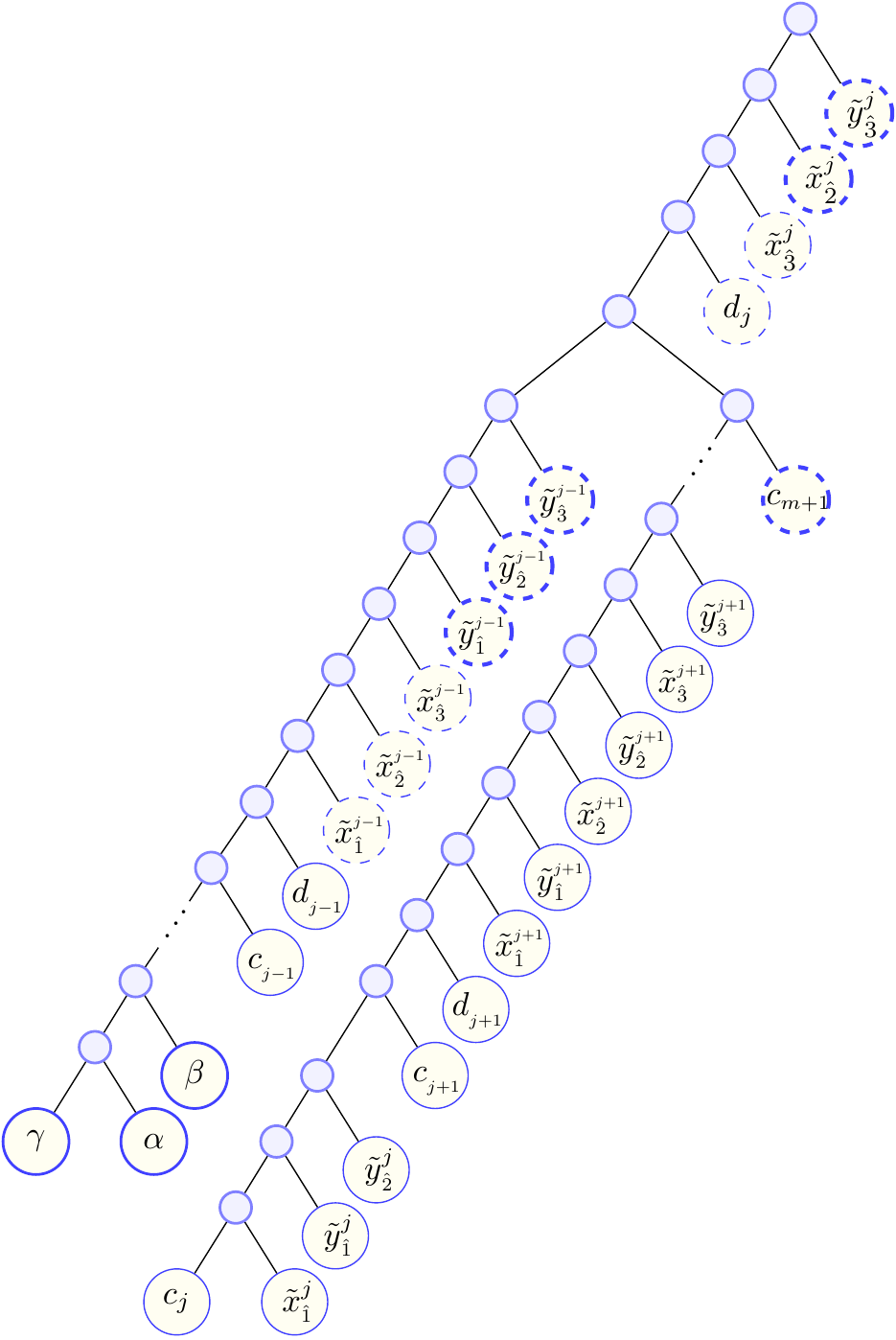} }
\else
\clipbox{7.5cm 1.15cm 6.25cm .5cm}{   \includegraphics[height=.73\textheight,angle=31.75]{rt-fig-cuttreeclause-arx.pdf} }
\fi
%\adjincludegraphics[height=5cm,clip,bb = 0 10 612 100,angle=30]{rt-fig-cuttreeclause-arx.pdf}
%    \fbox{\tikz{
%            \clip[use as bounding box] (0,0)++(.01\textwidth,.01\textwidth) rectangle ++(.45\textwidth,.45\textwidth);
%\draw (0,0) node[inner sep=0pt,name=micrograph]{\includegraphics[angle=10,width=1.6\textwidth]{rt-fig-cuttreeclause-arx.pdf}};
%\draw[ultra thick, white] (micrograph.south west)++(0.05*0.5\textwidth,0.05*0.5\textwidth)--++(0.5*0.25806\textwidth,     0)node[above,white,midway]{}};}
%}
\caption{$A^\times$ contains three arcs lying with clause $C_j$'s gadget.
%A tree consistent with $A'$, i.e., displaying all of $R'=\triples(A')$, but {\em not} displaying $\alpha \beta | \gamma$.
%The tree that can be constructed displaying all triples in $R'=\triples(A')$ but {\em not} $\alpha \beta | \gamma$, if $A'$ is acyclic and there is no path $P \subseteq A'$ from $\alpha\beta$ to $c_{m+1}\gamma$. 
%The root's left subtree contains all leaves appearing in A-nodes before the ``cut'' (plus $\gamma$), and its right subtree contains all leaves appearing after the ``cut'' (except $\gamma$).
\label{fig:cuttreeclause}
}
  \end{subfigure}
 }
%
%
%%\captionsetup{justification=centering}
%%\includegraphics[width=\linewidth, clip=true, trim={0cm 10cm 3cm 5cm}]{triplescaterpillar.jpg}
%%\input{triples-fig-tree}
%\includegraphics[height=.5\textheight]{triples-fig-pathtree.pdf}
%%\vskip .05cm
%\vskip -.1cm
%\caption{The caterpillar tree resulting from the triples corresponding to $P$, which thence satisfies $\alpha\beta|\gamma$.}
%\label{fig:caterpillar}
%\vskip .1cm
%\centering
%%\captionsetup{justification=centering}
%%\includegraphics[width=\linewidth, clip=true, trim={0cm 10cm 3cm 5cm}]{triplescaterpillar.jpg}
%%\input{triples-fig-tree}
%\includegraphics[height=.5\textheight]{triples-fig-nopathtree.pdf}
%%\vskip .05cm
%%\vskip -.05cm
\caption{Trees resulting from triples corresponding to an acyclic arc set $A' \subseteq A$ in $H$ that (a) do or (b,c,d) do not contain a path $P$ from $\alpha\beta$ to $c_{m+1}\gamma$, for different cases of $A'$ and $A^\times = A-A'$. Leaves corresponding to isolated A-nodes in the tree's Ahograph are shown with thick dashed borders, and the leaves corresponding to A-nodes that {\em become} isolated when those isolated A-nodes are removed are shown with thin dashed borders. In (b,c,d), leaves within the two major subtrees are also shown thus, with respect to those subtrees' Ahographs.
%The tree that can be constructed displaying all triples in $R'=\triples(A')$ but {\em not} $\alpha \beta | \gamma$, if $A'$ is acyclic and there is no path $P \subseteq A'$ from $\alpha\beta$ to $c_{m+1}\gamma$. The root's left subtree contains all leaves appearing in A-nodes before the ``directed cut'' (plus $\gamma$), and its right subtree contains all leaves appearing after the ``cut'' (except $\gamma$).
}
\end{figure}

%Let $H^* = (V_{\alpha \beta}^* \cup V_{c_{m+1}\gamma}^*, A')$, $H^\gamma = (V_{\alpha \beta}^* \cup V_{c_{m+1}\gamma}^* \cup \{ \gamma\alpha\}, A' \cup \arc(\gamma\alpha|\beta))$, $H^\times = (V_{\alpha \beta}^* \cup V_{c_{m+1}\gamma}^*, A' \cup A^\times)$, and $H^{\times\gamma} = (V_{\alpha \beta}^* \cup V_{c_{m+1}\gamma}^* \cup \{ \gamma\alpha\}, A' \cup A^\times \cup \arc(\gamma\alpha|\beta))$, and let $G^*$, $G^\gamma$, $G^\times$, and $G^{\times\gamma}$ be the corresponding Ahographs (respectively).
%
%
% mpj
%$G^*$ and $G^\times$ each have (at least) two components because no triple in either of the corresponding triple sets has $\gamma$ on the LHS, whereas $G^*$ may consist of three components, due to $\gamma$ and (potentially) to the cut $A^\times$. By definition, $H^\times$ contains an acyclic path $P$ from $\alpha\beta$ to $c_{m+1}\gamma$. Running BUILD on $\triples(P)$ will produce a caterpillar tree (see Fig.\ \ref{fig:pathtree}), via an Ahograph consisting of a tree and an isolated node $\gamma$ (see Fig.\ \ref{fig:ahograph}).
%
%Since $H^{\times\gamma}$ contains $P$ as well as $\arc(\gamma\alpha|\beta)$, $\gamma$ is not an isolated A-node in $G^{\times\gamma}$, which is connected, and so BUILD will fail when run on $G^{\times\gamma}$. 

Recalling the construction of $H$, there are three types of places where the absent cross-arcs $A^\times$ could be located: within a clause gadget, within a variable gadget, or elsewhere, i.e., {\em forced arcs} (viz., connecting arcs $a_1,...,a_4$ or arcs with tail of the form $c_j c_{j+1}$ following a clause $C_j$'s gadget). There is one special subcase, which we give a name to.

\begin{definition}
We call $A^\times$ {\em degenerate} if $A^\times$ lies within a variable $x_i$'s gadget, $|A^\times|=2$, and exactly one of its members has the form $\arc(b_ib'_i|\tilde x_i^1)$. (Its other member must by definition lie within the $x_i$ gadget's opposite side.) %Otherwise, $A^\times$ is {\em non-degenerate}.
\end{definition}

We deal with all cases besides an degenerate $A^\times$ in the following lemma.

\begin{lemma}\label{lem:treeexs}
%Let $A'$ be acyclic, suppose there is no path $P \subseteq A'$ from $\alpha\beta$ to $c_{m+1}\gamma$, and let $V_{\alpha \beta}^* \subseteq V_{\alpha\beta}$ and $V_{c_{m+1}\gamma}^* \subseteq V_{c_{m+1}\gamma}$ be defined as above,
%%in the proof of Lemma \ref{lem:notentailed} above, 
%and let $|V_{\alpha \beta}^*| = |V_{c_{m+1}\gamma}^*| \in [3]$, 
%%where is bijection between $V_{\alpha \beta}^*$ and $V_{c_{m+1}\gamma}^*$, with each pair corresponding to a unique A-edge in $A-A'$.
%which together yield $w \in [3]$ pairs $\{(u_h,v_h) : h \in [w]\}$, with each pair connected by an arc $a_h \in  A - A'$.
%
Let $R'$ be consistent.
%Assume that either $|A^\times|=1$ or $A^\times \cap \arcs(\{b_ib'_i|x_i^1,b_ib'_i\bar x_i^1\}) \ne 1$ for all $i \in [n]$, and that $R'$ is consistent.
%Suppose $A' \subseteq A$ is acyclic and 
Suppose there is no path $P \subseteq A'$ from $\alpha\beta$ to $c_{m+1}\gamma$, and that $A^\times$ is non-degenerate.
%either $|A^\times|=1$ or $A^\times \cap \arcs(\{b_ib'_i|x_i^1,b_ib'_i\bar x_i^1\}) \ne 1$ for all $i \in [n]$.
%
%Let the conditions of Lemma \ref{lem:uvpairs} hold, and suppose that either $|A^\times|=1$ or $A^\times \cap \arcs(\{b_ib'_i|x_i^1,b_ib'_i\bar x_i^1\}) \ne 1$ for all $i \in [n]$.
%do not lie within a variable gadget.
Then $R'$ does not entail $\alpha\beta|\gamma$.
%$R' \cup \{\alpha\gamma | \beta\}$ is consistent.
%there exists a tree displaying $R' \cup \{\alpha\gamma | \beta\}$.
\end{lemma}
\begin{proof}
%Now we will show that for such an $A''$, $R' = \triples(A'')$ is compatible with $\alpha\beta | \gamma$ not holding.
%
%Let $w \in [3]$ be the number of pairs $(u_h^*,v_h^*)$. 
%By Lemma, \ref{lem:uvpairs} there are $w \in [3]$ ordered pairs $\{(u_h,v_h) : h \in [w]\}$, with each pair connected (w.l.o.g.) by a unique arc $a_h \in A^\times \subseteq A - A'$.

First, suppose $A^\times$ contains any of the forced arcs, say $a_\times$, corresponding to a triple $p_\times q_\times|o_\times$. % and also to A-edge $\{p_\times,q_\times\}$ in $G^\times$.
%Arc $a_\times$ could be either a connecting arc, an arc in one side of a variable gadget (where none of the other side's arcs are included in $A'$, or an arc in one witness path of a clause gadget (where none of its other witness paths' arcs are included in $A'$). 
Now, consider the effect of the corresponding A-edge $\{p_\times,q_\times\}$'s absence from $G^\gamma$: $p_\times$ and $q_\times$ must lie in different components of it, since (observe) the only multi-edges in $G^\gamma$ are $\{b_i,b'_i\}$ for $i \in [n]$ and $\{c_j,d_j\}$ for $j \in [m]$, i.e., those due to nodes in $H$ with out-degree 2. Let $C^\alpha$ be the component in $G^\gamma$ containing $\alpha$ and $C^{\bar\alpha}$ the the other, and let the leaves be named so that $C^\alpha$ contains $p_\times$ and $C^{\bar\alpha}$ contains $q_\times$.

%Let the leaves be named so that $p_k$ is the leaf appearing on the RHS of the arc $a'=\arc(p'q'|o')$ in $H^\times$, i.e., $p_\times=o'$, in which case $q_\times$ does not appear in $\arc(e')$ and moreover A-node $q_\times$ does not appear in $C^\alpha$, and A-node $p_\times$ does not appear in $C^\alpha$. Therefore $G^\gamma$ will consist of two components: one before the cut (including $\gamma$) and one after, each of them containing an A-leaf whose A-edge label is an A-node lying in the {\em other} component: $q_\times=o'$ is an A-node in $C^{\bar\alpha}$, and $\gamma$ (the label of A-edge $\{c_mc_{m+1}\}$) is an A-node in $C^\alpha$.

Now, consider what happens when BUILD recurses on the triple and leaf sets corresponding to each of these components: in each case, BUILD will encounter a leaf set whose triples permit a linear ordering (similar to Table \ref{tbl:linord}), which therefore can be displayed in a (caterpillar) tree. In particular, in BUILD's recursive call for $[C^{\bar\alpha},V(C^{\bar\alpha})]$, $c_{m+1}$ will be an isolated node and the rest of the Ahograph will induce a DAG (directed acyclic graph) directed away from from $\gamma$; and in the recursive call for $[C^\alpha,V(C^\alpha)]$, $p_\times$ will be an isolated node and again the rest of the Ahograph will induce a DAG.
%therefore the Ahograph for that recursive call will consist of two components, an isolated A-node and a tree. 
In each of these two calls, therefore, BUILD's execution on its isolated-A-node-and-DAG pair will behave similarly ({\em despite the fact} that the portions of this Ahograph corresponding to clause gadgets may contain A-nodes and A-edges corresponding to all three witness paths, and the portions corresponding to variable gadgets may contain A-nodes and A-edges corresponding to both sides) to its execution on $\triples(P)$ for the acyclic $P$ in Lemmas \ref{lem:consistR} and \ref{lem:isentailed} (see Fig.\ \ref{fig:ahograph}), with each iteration isolating one additional node, thus constructing a caterpillar tree. These two caterpillars will become the root's child subtrees in the resulting tree displaying $R' \cup \{\gamma\alpha|\beta\}$ (see Fig.\ \ref{fig:cuttreeforced}).

Second, suppose there does not exist a path though the series of variable gadgets, i.e., from $b_1b'_1$ to $b_{n+1}b'_{n+1}$, meaning there does not exist a path through {\em some} variable gadget, i.e., from $b_ib'_i$ to $b_{i+1}b'_{i+1}$ for some $i \in [n]$. Then $A^\times$ must contain an arcs from both sides of $x_i$'s gadget, say $\arcs(z_i^\ell,z_{i^\ell+1}|z_i^{\ell+2},~ \bar z_i^{\ell'},\bar z_i^{\ell'+1}|\bar z_i^{\ell'+2}))$, where (because $A^\times$ is nondegenerate) either both or neither of $z_i^{\ell},\bar z_i^{\ell'}$ equals $b_i$. The absence of the two corresponding A-edges means that the Ahograph $G'$ will be disconnected. During BUILD's recursive calls for these two components, the whole process will play out similarly to how it did in the previous case (see Fig.\ \ref{fig:cuttreevar}).

%The third case is when the cut goes through some variable $x_i$'s gadget. Then $|A^\times|=2$ (if both sides of the $x_i$ gadget are including in $A'$). Since $A^\times$ is non-degenerate (i.e., either {\em both or neither} of these arcs are outgoing arcs of $b_ib'_i$), the Ahograph will be disconnected, and the whole process will play out similarly to the other two cases.

Third, suppose there does not exist a path through the series of clause gadgets, i.e., from $c_1d_1$ to $c_{m+1}d_{m+1}$. Since none of the forced arcs is missing, this means there does not exist a path through some clause gadget, i.e., from $c_jd_j$ to $c_{j+1}d_{j+1}$ for some $j \in [m]$. Then $A\times$ must contain an arc from each of the clause gadget's three witness paths (or similarly from {\em both} of them if the clause has only two literals).

%consider the case of the cut going through a clause gadget, in which case $A^\times$ consists of either two or  three arcs (since we have already dealt with the case of $|A^\times|=1$), each lying within one off that clause's witness paths. The case of two arcs is similar, so consider the most complex case of three.

Now there are several subcases to consider, but they all behave similarly. First, suppose $A^\times$ includes $c_jd_j$'s three outgoing arcs, each corresponding to a triple of the form $c_jd_j|\tilde x_{\hat w}^j$. Now, consider the effects of $c_jd_j$'s outgoing arcs' absence from $G^\gamma$: this means the three corresponding A-edges between $c_j$ and $d_j$ are absent, isolating $d_j$. Thus the initial input to BUILD consists of two components: isolated A-node $d_j$, and a DAG. In the recursive call for the DAG, the A-edge $\{c_{j-1},c_j\}$ labeled $d_j$ will no longer exist, and thus the corresponding Ahograph will consist of two components: one corresponding to the portion of the hypergraph from $\gamma$ through $C_{j-1}$'s gadget, and another corresponding top the portion from $C_j$'s gadget ({\em except for $d_j$}) though $c_{m+1}$. Now consider the recursive calls for these two components. Within the first component's recursive call, the three leaves of the form $\tilde y_{\hat w}^j$ will have been isolated, with the result that in subsequent calls the leaves of the form $\tilde x_{\hat w}^j$ will have been isolated, and so on, yielding a caterpillar tree. Similarly, $c_{m+1}$ will be isolated within the second recursive call, eventually yielding a second caterpillar tree. Then BUILD's final result will be a tree whose root's children are $d_j$ and a subtree whose two children are the roots of the two aforementioned caterpillar trees.

With three witness paths to cut, then there are  (up to isomorphism) either other subcases, which can all be verified through similar reasoning. The three cut arcs in any given case will be a subset from among the three $c_j \tilde y_{\hat w}^j$ nodes' outgoing arcs, the three $c_j \tilde x_{\hat w}^j$ nodes' outgoing arcs, and $c_jd_j$'s three outgoing arcs. Each cut arc outgoing from a $c_j \tilde y_{\hat w}^j$ node or a $c_j \tilde x_{\hat w}^j$ node will yield an isolated A-node. The initial Ahograph will therefore consist of either 1, 2, or 3 isolated A-nodes plus a DAG. Since it breaks ties arbitrarily, suppose BUILD always places isolated A-nodes as right children in the binary tree being constructed. Any time an A-node $\tilde y_{\hat w}^j$ is isolated (for some $w \in [3]$), in the following recursive call with $\tilde y_{\hat w}^j$ removed, $\tilde x_{\hat w}^j$ will be isolated. And in the recursive call following the isolation and removal of $\tilde x_{\hat w}^j$ for all $w=1,2,3$, $d_j$ will be isolated and removed.

At that point, the remaining Ahograph will consist of two DAGs $C^\alpha$ and $C^{\bar\alpha}$, where $C^\alpha$ contains $\alpha$ and $C^{\bar\alpha}$ does not. $C^{\bar\alpha}$ will also contain $c_j$ and any other A-nodes corresponding to $C_j$'s gadget that have not already been isolated and removed. In particular, $A^\times$ containing a cut arc outgoing from a $c_j \tilde x_{\hat w}^j$ node will mean $C^\alpha$ contains $y_{\hat w}^j$, and $A^\times$ containing a cut arc incoming to a $c_j \tilde x_{\hat w}^j$ will mean $C^\alpha$ contains $x_{\hat w}^j$ and $y_{\hat w}^j$.

The resulting binary tree (see Fig.\ \ref{fig:cuttreeclause}) will consist of a caterpillar region beginning at the root containing the 1-7 A-nodes isolated before the disconnection of the Ahograph into $C^\alpha$ and $C^{\bar\alpha}$, followed by a node whose two child subtrees are caterpillar trees resulting from $C^\alpha$ and $C^{\bar\alpha}$.
%
%The pattern is that the isolation of any A-node $y_{\hat w}^j$ will, in subsequent recursive calls, isolate the corresponding $x_{\hat w}^j$, the isolation of all three of which will in turn isolate $d_j$, after which the process continues as in the previous case, eventually yielding a tree (albeit a slightly differently shaped one).
\end{proof}

The problematic situation is when exactly {\em one} of the two arcs is outgoing from $b_ib'_i$. In this case, their absence deletes only one of the Ahograph's two A-edges between the pair $\{b_i,b'_i\}$, which does {\em not} disconnect the graph, meaning BUILD will fail.

We have been arguing that if a consistent $R'$ entails $\alpha\beta|\gamma$ then $\arcs(R')$ must contain an acyclic path from $\alpha\beta$ to $c_{m+1}\gamma$. Now we refine this to a slightly weaker (yet strong enough) implication: if a consistent $R'$ entails $\alpha\beta|\gamma$, then a slightly different consistent $R^+$ will too, and an acyclic path must exist within $\arcs(R^+)$.
%(even if it does not within $\arcs(R')$).

%We handle this case via the following argument.

Taking the contrapositive, we rewrite the previous lemma as:
\begin{corollary}
Let $R'$ be consistent.
Suppose $R' \vdash \alpha\beta|\gamma$ and that $A'$ contains no path from $\alpha\beta$ to $c_{m+1}\gamma$.
Then $A^\times$ is degenerate.
%Then $|A^\times|=2$ and $A^\times$ contains exactly one arc $a_\times$ of the form $\arc(b_ib'_i|\tilde x_i^1)$.
\end{corollary}

\begin{lemma}\label{lem:degeneratecutpath}
Let $R'$ be consistent.
Suppose $R' \vdash \alpha\beta|\gamma$ and that $A'$ contains no path from $\alpha\beta$ to $c_{m+1}\gamma$.
%Assume $|A^\times|=2$ and $A^\times$ contains exactly one arc $a_\times$ of the form $\arc(b_ib'_i|\tilde x_i^1)$, and that $R'$ is consistent.
%
%Suppose a consistent $R'$ entails $\alpha\beta|\gamma$, $|A^\times|=2$, and $A^\times$ contains exactly one arc $a_\times$ of the form $\arc(b_ib'_i|\tilde x_i^1)$.
Then there exists an acyclic path $P \subseteq A' \cup \{a_\times\}$.
\end{lemma}
\begin{proof}
%
%In this scenario, the two members of $A^\times$ (which are absent from $A'$) are one of the two possible arcs with tail $b_ibi_i$ (say, $\arc(b_ib'_i|x_i^1)$), and an arc within the $x_i$ gadget's other side (whose tail is not $b_ib'_i$). 
%
The corollary implies that $A^\times$ is degenerate.
%$|A^\times|=2$ and $A^\times$ contains exactly one arc $a_\times$ of the form $\arc(b_ib'_i|\tilde x_i^1)$.

Let $H^+$ be the hypergraph obtained by adding $a_\times$ to $H^*$, which by construction contains a path (say, $P^+$) from $\alpha\beta$ to $c_{m+1}\gamma$.
Let $R^+ = \triples(A^+)$, where $A^+ = A' \cup \{a_\times\}$ is $H^+$'s arc set.

We claim that $R^+$ is consistent. We know that the point at which BUILD would potentially fail when executed on a subset of $R$ is in a call when an A-edge of the form $\{c_j,\tilde y_{\hat i_j}^j\}{:}c_{j+1}$ has been removed and the Ahograph is not disconnected, because $\tilde y_{\hat i_j}^j$ is not isolated (see Fig.\ \ref{fig:ahocycle}). But if this occurs when BUILD is executed on $R^+$, it would also have occurred when BUILD is executed on $R^*$ because $\tilde y_{\hat i_j}^j$ would still be connected in that scenario as well (see Fig.\ \ref{fig:ahodiscordvar}).
Therefore $A^+$, and in particular $P^+$, is acyclic.
%
%This proves in the variable gadget case as well that if a consistent $R'$ entails $\alpha\beta|\gamma$, then there exists an acyclic path from $\alpha\beta$ to $c_{m+1}\gamma$, concluding the proof.
%
%
%We want to show that if $\triples(A' \cup A^\times)$ entails $\alpha\beta|\gamma$, then $\triples(A' \cup A^\times) \cup \{b_ib'_i|x_i^1\}$ does too.
%
%We claim that $H^{*\prime}$ is acyclic.
\end{proof}

This implies:

\begin{corollary}[$\Rightarrow$]\label{cor:ispath}
If there is a consistent $R'$ entailing $\alpha\beta|\gamma$ then there exists an acyclic path $P$.
\end{corollary}

Combining the Corollary \ref{cor:consentails} and \ref{cor:ispath} with Theorem \ref{thm:apdh}, we conclude:

\begin{theorem}
\entailprob\:is NP-Complete.
\end{theorem}

And because computing the closure reduces to deciding whether $R \vdash t$ for $O(|L|^3)$ triples $t$, we also have:

\begin{corollary}
\closureprob\:is NP-hard.
\end{corollary}

%Also:
%
%\begin{corollary}
%The complementary problem, of deciding whether a triple is entailed by {\em every} consistent subset of $R$ is co-NP-Complete.
%\end{corollary}
%\begin{proof}
%Solve the original problem for both negations of the triple.
%\end{proof}

%\subparagraph*{Acknowledgements.}

\section*{Acknowledgements}

This work was supported in part by NSF award INSPIRE-1547205, and by the Sloan Foundation via a CUNY Junior Faculty Research Award.
We thank Megan Owen for introducing the \ncnod\ problem to us.

\bibliographystyle{abbrv}
\bibliography{bib}

\end{document}